\documentclass{siamart171218}
\usepackage{color}
\numberwithin{equation}{section}
\usepackage{mathtools}
\usepackage{amssymb}
\usepackage{amsfonts}
\usepackage{graphicx}
\usepackage{epstopdf}
\usepackage{algorithmic}
\usepackage{enumerate}
\usepackage{accents}
\newcommand\thickbar[1]{\accentset{\rule{.9em}{1pt}}{#1}}

\newtheorem{rmk}{Remark}
\newtheorem{hyp}{Hypothesis}
\newtheorem{cor}{Corollary}

\numberwithin{rmk}{section}
\numberwithin{hyp}{section}
\numberwithin{cor}{section}

\def\ee{\end{equation}}
\def\bea{\begin{eqnarray}}
\def\eea{\end{eqnarray}}

\def\({\left(}
\def\){\right)}
\def\[{\left[}
\def\]{\right]}
\def\r|{\right\vert\right\vert}
\def\l|{\left\vert\left\vert}

\def\0{{\mathbf 0}}

\def\bw{{\mathbf w}}

\def\deg{\mathbf{deg}}

\def\bA{\mathbf A}

\def\bC{\mathbf C}
\def\bD{\mathbf D}

\def\bE{\mathbf E}

\def\bf{\mathbf f}
\def\bL{\mathbf L}

\def\bM{\mathbf M}
\def\bN{\mathbf N}
\def\bH{\mathbf H}

\def\bQ{\mathbf Q}

\def\bT{\mathbf T}
\def\bS{\mathbf S}

\def\bV{\mathbf V}

\def\be{\mathbf e}
\def\done{\mathbf {d1}}
\def\dtwo{\mathbf {d2}}

\def\b1{{\mathbf 1}}

\def\T{{\mathrm T}}

\newcommand{\TheTitle}{Structure- \& Physics- Preserving Reductions of Power Grid Models} 
\newcommand{\TheAuthors}{C. Grudzien, D. Deka, M. Chertkov, and S. Backhaus}

\headers{Structure- \& Physics- Preserving Reductions of Power Grid Models}{\TheAuthors}

\title{{\TheTitle}\thanks{Submitted to the editors July 10, 2017.
\funding{This work was supported by funding from the project REDDA of the Norwegian Research Council under contract 250711, and by funding from the U.S. Department of Energy's Office of DOE/OE Transmission Reliability Program as part of the DOE Grid Modernization Initiative.}}}

\author{
  Colin Grudzien \thanks{Nansen Environmental and Remote Sensing Center, Bergen Norway
    (\email{colin.grudzien@nersc.no}, \url{http://cgrudz.github.io}).}
  \and
  Deepjyoti Deka\thanks{T-4, Theory Division and the Center for Nonlinear Studies Los Alamos National Laboratory Los Alamos, NM 87545 (\email{deepjyoti@lanl.gov}).}
  \and
  Michael Chertkov\footnotemark\thanks{T-4, Theory Division and the Center for Nonlinear
Studies, Los Alamos Nationl Laboratory, Los Alamos, NM 87545 (\email{chertkov@lanl.gov}).}
  \and
  Scott N Backhaus\thanks{A Division of Los Alamos National Laboratory, Los Alamos, NM 87545 (\email{backhaus@lanl.gov}).}
}

\usepackage{amsopn}

\begin{document}

\maketitle

% REQUIRED
\begin{abstract}
The large size of multiscale, distribution and transmission, power grids hinder fast system-wide estimation and real-time control and optimization of operations. This paper studies graph reduction methods of power grids that are favorable for fast simulations and follow-up applications. While the classical Kron reduction has been successful in reduced order modeling of power grids with traditional, hierarchical design, the selection of reference nodes for the reduced model in a multiscale, distribution and transmission, network becomes ambiguous.  In this work we extend the use of the iterative Kron reduction by utilizing the electric grid's graph topology for the selection of reference nodes, consistent with the design features of multiscale networks.  Additionally, we propose further reductions by aggregation of coherent subnetworks of triangular meshes, based on the graph topology and network characteristics, in order to preserve currents and build another power-flow equivalent network.

Our reductions are achieved through the use of iterative aggregation of sub-graphs that include general tree structures, lines, and triangles.  Important features of our reduction algorithms include that: (i) the reductions are, either, equivalent to the Kron reduction, or otherwise produce a power-flow equivalent network; (ii) due to the former mentioned power-flow equivalence, the reduced network can model the dynamic of the swing equations for a lossless, inductive, steady state network; (iii) the algorithms efficiently utilize hash-tables to store the sequential reduction steps. The third feature allows for easy re-introduction of detailed models into the reduced, conceptual network, and makes the final reduced order model backward compatible with a sequence intermediate, partially reduced networks with varying resolution --- the ordered sequence of iterative reductions corresponds to a sequence of reduced order models. The performance of our graph reduction algorithms, and features of the reduced grids, are discussed on a real-word transmission and distribution grid.
We produce visualizations of the reduced models through open source libraries and release our reduction algorithms with example code and toy data.
\end{abstract}

% REQUIRED

% REQUIRED

\section{Introduction}

Power grids consist of the network of transmission and distribution lines connecting generators with end-users, enabling the transfer of electricity. The power grid of North America, in particular, is recognized as the most complicated machine built on earth \cite{greatest1, greatest2}. Topologically the grid is represented by a large, connected graph with nodes denoting buses (loads and generation) and edges representing lines. These nodes and edges are constructed in distinct formations across the physical scales in the problem of power delivery \cite{hines2010topological}. The transmission and distribution sub-networks exist in a hierarchical configuration, where the transmission sub-network consists of high voltage lines connecting generators to substations and the distribution sub-network connects substations to end users \cite{deka2016analytical}. System-wide monitoring and control of the grid involves simulation studies carried out by network authorities like independent system operators (ISO) \cite{huneault1991survey, kraning2014dynamic}. Simulating grid operations relies on accurate state estimation and optimization with respect to power-flow laws, describing interactions across layers of temporal and spatial resolution \cite{nishikawa2015comparative}.

Over time, the grid and its dynamical characteristics undergo changes with the introduction of new loads, generators and network components. The increased penetration of renewable energy, e.g., solar and wind power, has expanded the frontiers of the grid and also made issues regarding grid stability and control of paramount importance \cite{backhaus2013getting, von2013time, dvijotham2014storage}. Dynamic forcing from the distribution grid has historically been much smaller than the transmission components. For example, the amount of inertia and damping in the distribution grid are limited \cite{hoffman2006practical}. However, rooftop solar, the internet of things \cite{katz2017}, and other resources have cultivated the demand for \emph{decentralized} resource generation and control in the distribution sub-grid \cite{farhangi2010path, medina2010demand, dvijotham2012distributed,zhao2015distributed, deka2017structure}. With this demand comes the need for multiscale, dynamical models of the grid.

Owing to the large size and dense interconnections, the control, optimization and dynamical simulation of detailed grids faces implementation issues \cite{bienstock2014chance, bent2013synchronization}. Operational demands require reduced order and approximate schemes to improve the efficiency of computations and simulations of grid operations. However, one must ensure that the model reduction schemes are true to the original grid and have comparable dynamic behavior, or approximate the same.  For designing optimal power-flow and control schemes, it is common to study the transient stability of a reduced order model for the network in consideration --- transient stability in this case refers to the ability of the network to remain synchronous when subjected to large fluctuations in generation or faults in components.  If loss of synchrony appears due to transient instability, it is usually evident within two to three seconds of the initial disturbance \cite{kundur1994power}[see Chapter 13].

The dynamical behavior of the network in transient stability studies is often modeled in terms of swing equations under the assumptions of: (i) purely imaginary line admittances; (ii) lossless power-flow; (iii) constant active and reactive power at load buses;  (iv) constant voltage magnitudes at all generators; and (v) each generator rotor frequency is sufficiently close to the fixed, operating frequency.  The swing equations describe the evolution of perturbations to the rotors' frequency from the steady state on the time scale of seconds.  Particularly, the dynamic swing equations determine the linear stability of the power-flow over this short time scale \cite{simpson2017theory,lokhov2018online}.  Transient stability studies end at the time scale of tens of seconds in which the above assumptions may not hold.  While the swing equations are a substantial reduction to the grid physics, they provide a critical analysis of whether the generators can maintain synchrony at the operating frequency in realistic physical conditions, or if the dynamics will become unstable due to the loss of synchrony in its configuration on short time scales.      

In the case of a lossless, inductive, steady state network, which we will refer to as a transient stability regime, the Kron reduction can be used to produce a reduced order, electrically equivalent model for the network's power-flow \cite{dorfler2010synchronization}.  However, while the synchronization analysis becomes tractable for the reduced order model, authors stress that the direct representation of the synchronization conditions for the full network are often lost \cite{dorfler2013kron}[see section G].  Despite the limitations of this approach, it allows for a physically consistent, and computationally feasible, analysis of the full network and its optimal power-flow and control. Likewise, although it is no longer an electrically equivalent network, the above techniques are commonly used for qualitative study of the dynamics of mid-term stability, with proper modifications and adequate representation of grid dynamics.   Mid-term stability refers to large frequency and voltage deviations, and the network response, on the order of tens of seconds to minutes.  In this intermediate time scale, the dynamics will become increasingly nonlinear due to the dynamic simulation of loads, generation and network and generator control actions, necessary to simulate the actual grid response \cite{kundur1994power}[see Chapter 16]. 

This paper analyses system-aware graph reductions of large power grids, to construct conceptual networks amenable for follow-up action and reanalysis. System-awareness here refers to the use of the graph topology and of parameters such as nodal voltage in the procedure --- by design, we perform our reductions in such a way as to be consistent with the iterative Kron reduction, or otherwise to produce another power-flow equivalent network within the transient stability regime.  We aim to preserve topological features such as presence of graph paths, and particularly the graph's sparsity which is not guaranteed in the Kron reduction process.  Generically, reducing a node with the Kron reduction will replace said node with new lines, forming links between all other nodes to which the reduced one was connected.  As a simple example, reducing the nodes $\{b_1,b_2,b_3\}$ in the left hand side of \cref{fig:krondense} via the Kron reduction produces a complete graph on the right hand side.  Preserving topological characteristics of the network is necessary for our underlying goal which is to develop graph reduction schemes that preserve qualitative features of the original grid's structure and dynamic behavior in the transient stability regime, for accurate state estimation, disturbance prediction and distributed control schemes.

\begin{figure}[ht]\label{fig:krondense}
\center
\includegraphics[width=.7\linewidth]{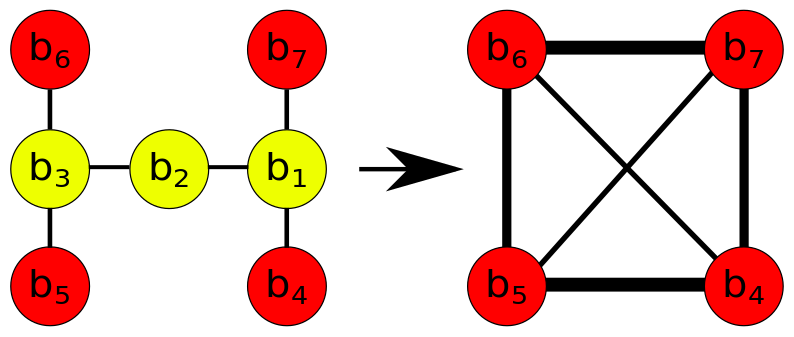}
\caption{The Kron reduction of nodes $\{b_1,b_2,b_3\}$ on the left hand side produces the complete graph on the right hand side.}
\end{figure}

\subsection{Contribution}
\label{sec:main}

There is extensive research into reduction algorithms for improving the analysis of large networks and reducing the computational burden therein. Community detection approaches use graph based methods to collapse sub-networks into smaller, representative components, e.g. Kannan et. al. analyze criteria for effective clustering approaches in relation to the spectrum of the graph Laplacian \cite{kannan2004clusterings}; Newman develops reduction methodology in terms of the node-group connectivity measure of modularity \cite{newman2004fast}. While community detection methods have applications in power systems, these approaches are not in of themselves appropriate for constructing a dynamically consistent reduced order model.  The work of Huo \& Cotilla-Sanchez seeks to \cite{huo2018power} preserve dynamical features of clustered communities by scoring the clusters based on power flow characteristics and applying an evolutionary algorithm.  Other works in circuit design have focused on network reductions which preserve static power-flow computations, e.g. Zhou et. al. study block based hierarchical graph reduction schemes for fast solution to power-flow for in-chip circuits \cite{reduction2}; Wang provides a deterministic random walk based pre-processing and graph reduction algorithm also aimed at solving the DC power-flow problem \cite{reduction1}; Chen \& Chen present Krylov-subspace iterative methods for preconditioning \cite{reduction4}.

Notable power systems reduction methodologies, designed to preserve the dynamical characteristics of the network, include the slow coherence techniques of Chow \& Kokotovic \cite{chow1985time,chow2013power} and the classical Kron reduction of Gabriel Kron \cite{kron1939tensor}.  Slow coherency utilizes the underlying structure of large power grid networks which contain subnetworks of weakly coupled and strongly coupled coherent groups of nodes.  In regional power grid networks, there are large load centers in big cities served by large generating stations, often far away from the load, utilizing high voltage transmission systems.  However, for practical purposes in balancing load demands, regional operating authorities exchange power between regions using weakly coupled, inter-area lines to share base load and reserves.  Coherent sets emerge physically as strongly coupled regional areas which have dense interconnection compared to the weak coupling formed by sparse interregional lines, typically with higher impedance or heavy loading.  

However, while the interregional coupling may be weak on a fast time scale, the long term dynamical behavior of inter-area machines is often strongly coupled on a long time scale.  Slow coherency is the phenomena in which groups of machines form coherent sets interregionally on long time scales, swinging against each other at oscillatory frequencies slower than the frequencies of machines within the a single densely connected region.  Slow coherence based reductions construct reduced order models by aggregating nodes within coherent sets formed in fast time scales, and deriving the reduced order dynamics from the underlying fast-slow, time scale separation \cite{chow2013slow}.
        
The Kron reduction has been applied extensively in power systems analysis with success in control and optimization problems, and D\"{o}rfler \& Bullo in particular, provide a detailed mathematical analysis of the classical Kron reduction for the use in control and monitoring of smart grids \cite[and references therein]{dorfler2013kron}. Given a choice of reference nodes, the Kron reduction uses Gaussian elimination to pare down the full network to a reduced model that is power-flow equivalent from the perspective of the references.  The selection of reference nodes is unambiguous for transmission networks under a classical, hierarchical distribution design. However, the deployment of decentralized generation and storage in the distribution sub-grid makes the selection of the references problematic: while individual distribution nodes do not provide significant generation, the aggregation of these can strongly impact the optimal power-flow and control problem. 

In order to utilize the Kron reduction analysis for a multiscale, distribution and transmission network, we propose graph topological methodology to select the reference nodes in the iterative reduction of the network.  Similar to the motivation of slow coherence approaches, we seek to exploit the underlying structure and design of multiscale electric grid networks to inform our choices when aggregating nodes into representative, but simplified models.  Utilizing the network topology, and the typical electric grid design features which characterize this topology, we extend the use of the iterative Kron reduction, automatically selecting reference nodes in such a way as to respect the dynamics and qualitative features of multiscale, distribution and transmission, electric grid networks.  To produce further reductions beyond the Kron reduction, we utilize the network topology to aggregate topologically coherent sets of nodes of similar voltage, to produce power-flow equivalent, reduced order models in the transient stability regime.

Our key contribution is developing a series of sequential and invertible graph reduction algorithms, and demonstrating the viability of these techniques on a real, electric grid in a US Midwest utility. Our methodology emphasizes three design features: (i) the reductions are system aware, respecting the network topology of multiscale electric grids; (ii) they are power-flow consistent in the sense that each sequential map is either equivalent to a step in of the iterative Kron reduction, or otherwise aggregates nodes into a power-flow equivalent network; and (iii) sub-sequences of the iterative reductions can be inverted, to produce intermediate resolution models for the network.  System-awareness enforces that the reduced network respects the power-flow of the full network, but the invertibility of the reduction allows users to give complex dynamical features increased resolution by inverting the nodal aggregation post-facto.  The ordered sequence of reductions corresponds to a sequence of reduced order, power-flow equivalent models, which represent the network features at intermediate scales.  By utilizing the graph topology, and by maintaining network characteristics of graph paths and nodal voltage thresholds, our reductions furthermore preserve the strong and weak coupling of coherent subsystems, present in large scale power grids.  We therefore suggest, though it goes beyond the scope of the work, that our topologically based reductions are compatible with the dynamical properties of slow coherency and our reduced network may be further post-processed by these techniques.

Unstructured data, describing the placement of clustered nodes in the final reduced model, can be used to parameterize net power-flow.  Specifically, the nodes present in the final reduction can be used as reference nodes for the Kron reduction, and/or to describe simple nodal aggregation and its respective power-flow.  However, this unstructured data is insufficient to increase the resolution on a specific nodal cluster.  Effective implementation of data structures, tracking the sequence of reductions, has been an integral component of our work: our design enforces backwards compatibility, with intermediate, partially reduced representations of network features --- by tracking the sequence of reductions, one can select a different set of reference nodes in the sequence of reduced networks to produce an intermediate scale model.  Implementing these techniques on a real, multiscale electric grid, we present the results of our analysis, studying the graph characteristics of reduced networks. We interactively visualize the reduced network, likewise utilizing graph-topology, rather than geographical location to qualitatively analyze the results.

In \cref{sec:alg}, we present our main results, including our reduction algorithms and the analysis of their performance on the real, multiscale network. In \cref{sec:visualization} we demonstrate conceptual visualizations the reduced, case study network, and the nodal clustering produced by the algorithms. We detail our use of data structures in \cref{sec:data}, explaining how to invert the algorithms to increase the resolution post-facto. Utilizing the reduced network for on-line, dynamic modeling of multiscale electric grids is discussed in \cref{sec:conclusions} where we introduce future directions of research. Finally, example code and test data are available as supplementary material on-line \cite{grudziengithub2017} with interactive visualizations available in web browsers \cite{grudzienpersonalpage}.

\subsection{Notations and preliminaries}
In the following, we draw on the work of D\"{o}rlfer \& Bullo \cite{dorfler2013kron} to define the graph Laplacian and the loopy Laplacian (nodal admittance matrix), utilized in computing the power-flow of an electric grid network, and the equivalent power-flow for the \emph{Kron reduced} model.  The networks under consideration in this work will be understood in terms of algebraic, connected graphs, with a \textbf{node set} denoted $\bN$, of order $\vert \bN\vert = n<\infty$, and \textbf{edge set} $\bE \subset \bN \times \bN$.  Each node $b_i \in \bN$ will represent a generator or load bus in a regional, multiscale electric grid, with nominal voltage $v_i \in (0,1000)$ kilovolts.  We will identify each node $b_i$ with its index $i$ interchangeably in the text.  Edges in the network are undirected, giving rise to a symmetric \textbf{adjacency matrix}, $\bA\in \mathbb{C}^{n\times n}$.  The adjacency matrix, used to describe the power-flow, includes self loops, i.e., $A_{ii} \neq 0$, representing the shunt admittance at the bus $b_i$.  The shunt admittances describe loads in the network drawing a current.  The non-diagonal elements of $\bA$, $A_{ij} = A_{ji}$, denote line admittances used to describe the power-flow between nodes $b_i$ and $b_j$.  The injections and demands of \textbf{currents} are represented by a vector $\bC \in \mathbb{C}^{n}$, while \textbf{nodal voltages} are described in a vector form by $\bV \in \mathbb{C}^n$.
\begin{hyp}\label{hyp:inductive}
Assume that the adjacency matrix $\bA \in \mathbb{C}^{n\times n}$ defines a connected graph.  Moreover, assume that all non-zero entries of the adjacency matrix $\bA$ are inductive, i.e., that they are pure-imaginary and negative, and that at least one diagonal element $A_{ii}\neq 0$ for $1 \leq i \leq n$.
\end{hyp}

For an arbitrary $n\times n$ matrix $\bM$, we denote the entry in the $i$-th row and $j$-th column equivalently as $M[i,j]=M_{i,j}$.  We will define the \textbf{weighted degree} matrix $\bD$, the \textbf{graph Laplacian} $\bL$ and the \textbf{loopy Laplacian} $\bQ$, such that
\begin{align}
\bD &\triangleq {\rm diag} \(\left\{\sum_{k=1}^n A_{i,k}\right\}_{i=1}^n\),\label{eq:degreestrength} \\
\bL &\triangleq \bD - \bA \label{eq:laplacian},\\
\bQ &\triangleq \bL + {\rm diag} \(\left\{A_{i,i}\right\}_{i=1}^n\).\label{eq:loopy}
\end{align}
\begin{rmk}
The additional presence of self loops in the adjacency matrix $\bA$ can instead be used to model an equivalent, augmented circuit including a ground node; this produces an augmented $(n+1) \times (n+1)$ Laplacian, where all the self loops are attached to the ground defined in terms of the sum of all shunt admittances in the adjacency matrix \cite{dorfler2013kron}.
\end{rmk}   

\begin{lemma}\label{lemma:invertibleadjacency}
If the adjacency matrix satisfies \cref{hyp:inductive} then the loopy Laplacian is invertible.
\end{lemma}
\begin{proof}
Due to the connectivity of the graph, the matrix $\bQ$ is irreducible \cite{dorfler2013kron}.  But clearly, $\bQ$ is also diagonally dominant, with at least one diagonal element strictly dominant.  By Corollary 6.2.27 of Horn \& Johnson \cite{horn1990matrix}, $\bQ$ is invertible.
\end{proof}

The above defined loopy Laplacian describes the classical matrix of \textbf{nodal admittances}.  The diagonal elements of the loopy Laplacian are defined to be the \textbf{self-admittances}, equal to the sum of all admittances terminating at the associated node.  The off diagonal elements in the loopy Laplacian are equal to the negative of the associated line admittances.  Using the loopy Laplacian in equation \cref{eq:loopy}, we define the \textbf{current balance equations} as the matrix form of Ohm's law,
\begin{align}
\bC = \bQ \bV.\label{eq:currentbalance}
\end{align}
Likewise, we define the \textbf{power-flow equations} as
\begin{align}
\bS =  \bV \circ \thickbar{\bC} \label{eq:powerflow}
\end{align}
where $\circ$ is the Hadamard product, and $\thickbar{\bM}$ is the complex conjugate of the matrix $\bM$.  The vector $\bS$ is defined as the vector of \textbf{power injections}.  The sum of all power injections is defined to be the \textbf{net power}.
\begin{hyp}\label{hyp:lossless}
We will assume that the power-flow is lossless and net power is balanced, 
\begin{align}
\sum_{j=1}^n S_j = 0,
\end{align}
i.e., the sum of all the power injections is equal to zero.
\end{hyp}
\begin{rmk}
For steady state dynamics studies involving small injection fluctuations, the majority of which are at high voltage unreduced nodes, this hypothesis is a common approximation, if not precisely satisfied.
\end{rmk}

The loopy Laplacian $\bQ$ has the elementwise definition,
\begin{align}
Q_{i,j} = 
\begin{cases}
-A_{i,j} & {\rm if}\hspace{2mm} i \neq j\\
\sum_{k=1}^n A_{i,k}& {\rm if}\hspace{2mm} i=j,
\end{cases}\label{eq:loopyelement}
\end{align}
such that
\begin{align}\begin{split}
Q_{i,i} &= \sum_{k=1}^n A_{i,k}\\ 
       &= A_{i,i} -\( \sum_{k \in \{1,\cdots,n\} \setminus \{i\} } Q_{i,k}\).
\end{split}\label{eq:qdiag}
\end{align}
Therefore, we can always recover the adjacency matrix of a graph (and thus the full graph) from the associated loopy Laplacian via
\begin{align}
A_{i,j} = 
\begin{cases}
-Q_{i,j} & {\rm if}\hspace{2mm} i \neq j\\
\sum_{k=1}^n Q_{i,k}& {\rm if}\hspace{2mm} i=j.
\end{cases}\label{eq:adjacencyelement}
\end{align}

Given an arbitrary $n\times n$ matrix $\bM$, and some index set $\alpha = \{1, \cdots, m\}$ where $1<m<n$, we decompose the matrix $\bM$,
\begin{align}
\bM = 
\begin{pmatrix}
\bM_{[\alpha, \alpha]} & \bM_{[\alpha, \alpha)} \\
\bM_{(\alpha, \alpha]} & \bM_{(\alpha, \alpha)}
\end{pmatrix},\label{eq:matrixdecomp}
\end{align} 
such that $\bM_{[\alpha, \alpha]}, \bM_{[\alpha, \alpha)}, \bM_{(\alpha, \alpha]}, \bM_{(\alpha, \alpha)}$ are of dimensions $m \times m$, $m \times (n-m)$, $(n-m)\times m$ and $(n-m) \times (n-m)$ respectively. 
For an arbitrary $n\times 1$ vector $\bw$, we similarly define $\bw_{[\alpha]}$ to be the sub-vector of the elements indexed by $\alpha$ and $\bw_{(\alpha)}$ to be the sub-vector of elements indexed by the elements of $\{1, \cdots , n\} \setminus \alpha$.
 
Using the above operators, we can define the \textbf{Kron reduction} abstractly via the Schur complement of the loopy Laplacian $\bQ$ with respect to a sub-matrix corresponding to the nodes to be reduced.  In particular, define the index set $\alpha = \{1, \cdots , m\}$ such that $1 < m < n$, corresponding to a set of reference nodes --- we will denote the complementary index set $\{m+1, \cdots , n\}$ the interior nodes.   Then, the \textbf{Kron reduced loopy Laplacian} is given by
\begin{align}\label{eq:kronll}
\bQ^{\rm red} \equiv \bQ/\bQ_{(\alpha, \alpha)} \triangleq \bQ_{[\alpha, \alpha]} - \bQ_{[\alpha, \alpha)} \bQ_{(\alpha,\alpha)}^{-1} \bQ_{(\alpha, \alpha]},
\end{align}
where the notations for the sub-matrices are defined in equation \cref{eq:matrixdecomp}.  Note, up to re-indexing the nodes in the network (and associated shifts in the adjacency/ Laplacian matrices), the above reduction can be performed with respect to any $\alpha \subsetneq \{1,\cdots,n\}$, $\vert \alpha \vert > 1$, provided $\bQ_{(\alpha, \alpha)}$ is nonsingular.

Using equation \cref{eq:adjacencyelement} we see that the Kron reduced network's adjacency matrix can be reconstructed from the reduced loopy Laplacian $\bQ^{\rm red}$, and D\"{o}rfler \& Bullo \cite{dorfler2013kron} prove that this reduction is well defined.  Moreover, the authors show that the \textbf{Kron reduced current balance and power-flow} equations are given as
\begin{align}
\begin{split}
\bC^{\rm red} & =\bQ^{\rm red} \bV_{[\alpha]}\\
 & =\bC_{[\alpha]} + \bQ^{\rm ac} \bC_{(\alpha)}  
\end{split}\label{eq:currentbalancered} \\
\bS^{\rm red}& = \bV_{[\alpha]} \circ \thickbar{\bC}^{\rm red}, \label{eq:powerflowred} 
\end{align}
respectively, where the \textbf{accompanying matrix} $\bQ^{\rm ac}$ 
\begin{align}
\bQ^{\rm ac} &\triangleq - \bQ_{[\alpha,\alpha)} \bQ^{-1}_{(\alpha, \alpha)}\label{eq:qaccompany}
\end{align}
maps the Kron reduced, internal currents to the reference nodes, and $\bS^{\rm red}$ is defined as the \textbf{reduced power injection vector}.

To study the pure network topology of the power grid, we introduce a topological connectivity matrix, comprised entirely of ones and zeros, which excludes self loops present in the adjacency matrix and normalizes all the line admittances off the principal diagonal.  Specifically, the \textbf{topological connectivity matrix} $\bT$ is defined in terms of the adjacency matrix $\bA$, elementwise via
\begin{align}
T_{i,j} &\triangleq 
\begin{cases}
1 & {\rm if} \hspace{2mm} A_{i,j} \neq 0 \hspace{2mm} {\rm and} \hspace{2mm} i\neq j\\
0 & {\rm if} \hspace{2mm} i=j
\end{cases}
\end{align}
For each node $b_i\in \bN$, we will define the \textbf{topological degree} 
\begin{align}
\deg(b_i)\triangleq \sum_{k=1}^n T_{i,k},
\end{align}
equal to the number of nodes $b_i$ is connected to within the network, excluding self loops and line parameters.  The graph density is defined by
\begin{align}
d = \frac{ 2 \vert \bE \vert}{\vert \bN \vert ( \vert \bN \vert -1 )} ,
\end{align}
such that $d\in [0,1]$ with smaller values describing sparsely connected graphs, and a value of $d=1$ corresponds to a complete graph where all nodes share an edge.

Using the topological degree of each node in the network, and the graph paths which define it, we will utilize simple graph searches to select the reference nodes for the iterative Kron reduction.  The graph topological reductions are inspired by the design structure of multiscale power grids, which embed subgraphs composed of generalized tree structures, strings of transmission nodes, and triangular configurations.  To show the equivalence of our degree one and degree two graph topological approach to the Kron reduction, we will extensively utilize several important properties of the Kron reduction, proven in D\"{o}rfler \& Bullo's work \cite{dorfler2013kron}.
\begin{theorem}\label{theorem:iteration}
Let the loopy Laplacian $\bQ$ define a network of $n$ nodes, with associated the adjacency matrix $\bA$.  Assume the network satisfies \cref{hyp:inductive} and \cref{hyp:lossless}, and let $\alpha =\{1,\cdots, m\}$,  $1< m<n$, be some index set.  Then the following hold: 
\begin{enumerate}[(i)]
\item The Kron reduced loopy Laplacian, $\bQ^{\rm red}$ as defined in equation \cref{eq:kronll}, exits for any such $\alpha$.
\item Let $\beta = \{1, \cdots , p\}$ be any index set such that $1< m < p < n$.  Then the Kron reduced loopy Laplacian with respect to $\beta$, denoted $\bQ^\beta \triangleq \bQ/ \bQ_{(\beta, \beta)}$, exists and
\begin{align}
\bQ/\bQ_{(\alpha, \alpha)} =\bQ^\beta / \bQ^\beta_{(\alpha, \alpha)}.
\end{align}
That is, the Kron reduction with reference nodes defined by $\alpha$ can be produced iteratively: (1) first applying the Kron reduction with respect to an arbitrary superset of reference nodes $\beta$; (2) secondly computing the Kron reduction of this reduced network, $\bQ^\beta$, with respect to the index set defined by $\alpha$.
\item If $\bA$ satisfies \cref{hyp:inductive}, then the Kron reduced adjacency matrix $\bA^{\rm red}$ satisfies \cref{hyp:inductive}.
\item The reduced power injections under the Kron reduction, defined in equations \cref{eq:currentbalancered} and \cref{eq:powerflowred} are also lossless, and net power is preserved, thus satisfying \cref{hyp:lossless}. 
\end{enumerate}
\end{theorem}
\begin{proof}
Statement $(i)$ above is the \textbf{existence property} of lemma II.1, and statement $(ii)$ above is a simple corollary of the \textbf{quotient property} in lemma III.3 \cite{dorfler2013kron}.  Statement $(iii)$ can be understood from the \textbf{monotonicity property} in theorem III.6, the \textbf{closure property} in lemma II.1 and the \textbf{closure of irreducibility property} in theorem III.6 \cite{dorfler2013kron}.  

Under the assumption that all elements of $\bA$ are real and non-negative, D\"{o}rfler \& Bullo show that the elements of the Kron reduced adjacency matrix are monotonically increasing, such that
\begin{align}
 A^{\rm red}_{i,j} \geq  A_{[\alpha, \alpha]}[i,j].
\end{align} 
If the elements of $\bA$ are pure imaginary and non-positive, we will define $\widetilde{\bA} \triangleq i \bA$, which has real, non-negative entries.  From equation \cref{eq:loopyelement}, we find that the loopy Laplacian associated to $\widetilde{\bA}$ is given identically by $\widetilde{\bQ} \triangleq i \bQ$.  Therefore, the Kron reduced loopy Laplacian is given 
\begin{align}
\widetilde{\bQ}^{\rm red} &= \(i \bQ_{[\alpha, \alpha]}\)  - \(i\bQ_{[\alpha, \alpha)} \)\(i \bQ_{(\alpha,\alpha)}\)^{-1}\(i \bQ_{(\alpha, \alpha]}\),\\
&= i \bQ^{\rm red}.
\end{align}
In particular, from equation \cref{eq:adjacencyelement}, we find that the Kron reduced adjacency matrix defined by $\widetilde{\bQ}^{\rm red}$ is given identically by $\widetilde{\bA}^{\rm red} \triangleq i \bA^{\rm red}$.  By the monotonicity property in III.6, we know
\begin{align}
\widetilde{A}^{\rm red}_{i,j} &\geq  \widetilde{A}_{[\alpha, \alpha]}[i,j], & \Rightarrow & & \frac{1}{i} A^{\rm red}_{i,j} &\leq \frac{1}{i}A_{[\alpha,\alpha]}[i,j]
\end{align}
so that the elements of $\bA$ monotonically decrease along the imaginary axis under the Kron reduction, and therefore $\bA^{\rm red}$ is inductive.  

The closure property in lemma II.1 determines that at least one diagonal element of $\bA^{\rm red}$ is non-zero if and only if at least one diagonal element of $\bA$ is nonzero.  Moreover, $\bQ$ is irreducible if and only if $\bQ^{\rm red}$ is irreducible, such that the reduced network is connected.  Therefore, the Kron reduced network satisfies \cref{hyp:inductive}.   Finally statement $(iv)$ is true by construction.  The reduced network has lines given by combinations of lines in the full network, which are thus also lossless.  The Kron reduction is equivalent to Gaussian elimination of voltages of the interior nodes in equation \cref{eq:currentbalance}.  However, the usual power-flow equations \cref{eq:powerflow} do not apply to the reduced network, as they do not preserve the net power in the reduced network.  In particular, it is the accompanying matrix in equation \cref{eq:qaccompany} that makes the correction in the currents, mapping the reduced currents to the reference nodes, and preserves the net power in the vector of reduced injections \cref{eq:powerflowred}.
\end{proof}

\begin{figure}\label{fig:notations}
\small
\center
\begin{tabular}[h]{| l | l || l | l || l | l | }
\hline
$\mathbf{N}$ & node set & $\mathbf{E}$ & edge set & $\mathbf{H}$ & hash table \\
\hline
$\mathbf{d1}$ & degree one reduced & $\mathbf{d2}$ & degree two reduced & $\mathbf{tri}$ & triangle reduced \\
\hline
$\mathbf{vThr}$ & voltage threshold & $\mathbf{dThr}$ & degree threshold & $\mathbf{deg}(b)$ & degree of $b$ \\
\hline
\end{tabular}
\caption{Algorithm notations, e.g. the degree one reduced node and edge sets are denoted $\mathbf{d1N}$, $\mathbf{d1E}$ respectively. The degree one reduction data is stored in $\mathbf{d1H}$.  Degree and voltage thresholds are criteria set for greedy triangular reductions in \cref{alg:triangular}.} \end{figure}

To better understand how our reductions change the structure of the network, we will compare the degree distributions of the reduced networks.  A common method to measure the distance between two probability distributions with different support is the first Wasserstein distance, also known as the earth mover's distance \cite{rubner1998metric}.  When there is an underlying distance in the outcome space, the first Wasserstein distance is defined as the minimal work in the optimal transport problem, moving the ``weight'' of one distribution into the other.  In the discrete distribution case, the work is computed as the sum of all weighted distances between all pairs of bins, where each distance is weighted by how much probability is transfered from one bin to the other. Computationally, this distance is computed by linear optimization, and we use the code provided by Rubner et al. \cite{rubner1998code} to compute the distance between the degree distributions for our reduced networks.  In this case, the underlying distance is given by the L1 distance between the degrees in $\mathbb{N}$. 

\section{Reduction algorithms}
\label{sec:alg}
This section develops our algorithms for the sequential reduction of a multiscale, distribution and transmission, power grid network --- a reference for our notations can be found in \cref{fig:notations}.  Our methodology is designed to achieve the following objectives:
\begin{enumerate}
 \item the reduced network preserves qualitative features of the original grid's structure and dynamic behavior in the transient stability regime;
 \item the full network can be fully or partially reconstructed from the reduced model, at varying, intermediate resolutions;
 \item the reduced network is of a scale that is amenable to interactive visualization.
\end{enumerate}
We focus on off-line methods meeting the above objectives. As an off-line reduction, we use static network characteristics to produce a model which is robust to changes in the production, consumption, and dynamic characteristics of loads and generators. In particular, the same reduced network will be used to represent different load and parameter states including the inertia, damping and frequency control settings. Any such reduction implicitly assumes infrequent updates. These updates occur when topology and parameter estimations (line impedances, transformer settings, and related) undergo a significant change, e.g. in seasonal transitions or following major network modifications. With such schemes, operators infrequently need to reproduce reduced order models (and their associated parameterizations). However, the scheme can be improved with state estimation available through phasor measurement unit (PMU) technology in real time \cite{phadke2002synchronized}, i.e. on-line. The construction of a robust reduction scheme involving on-line state estimation will be the subject of future research. In the remainder of this section we discuss four consecutive sub-steps of our off-line scheme and analyze the performance of the algorithms on a real multiscale network.

\subsection{Degree zero reductions}

In our following algorithms, we will always assume that the network under consideration is formed by a graph with a single connected component, and therefore, has irreducible adjacency, Laplacian and loopy Laplacian matrices.  We will also assume that the vector of nodal voltages $\bV$ is also entirely non-zero.  With our test network this means that, as a preliminary step, we remove nodes of zero nominal voltage and restrict the test network to a single connected component, which will represent the full regional power grid.  In our test case, the connected sub-network contains 53,155 nodes, 63,832 edges, 268 PMU devices and 4,332 generators. The distribution of node degrees is given in \cref{fig:d1_d2_dist}. The mean node degree is 2.40 with standard deviation of 1.61, and a max degree of 40.  The graph density is approximately $4.518\times 10^{-5}$.  In our algorithms, the analysis is performed by manipulating the topological connectivity matrix $\bT$, but we describe the reduction algorithms at a high level with the node and edge sets, $\mathbf{N}$ and $\textbf{E}$.

\subsection{Degree one reductions}
\label{section:degreeone}
While the physical lines constituting the distribution sub-network form meshed, loopy graphs, the operational topology for load balancing consists of radial tree structures \cite{deka2017structure}. Operational switching disconnects lines and the meshed topology so that the substations, connected to the transmission network, form roots of disjoint trees in the distribution sub-networks. This structure differs significantly from the transmission network which typically has multiple loops energized at all times to guarantee continuous delivery to the substations \cite{deka2016estimating, deka2017structure}. In an operational window where the radial structure is unchanged, the distribution graph structure lends itself to an intuitive representation of the network's multiscale coupling. We map disjoint distribution trees to their respective roots at substations --- in our conceptual network, these terminal roots form super nodes which are used to represent the entire behavior of the tree.

\begin{figure}[ht]\label{fig:treecollapse}
\center
\includegraphics[width=.9\linewidth]{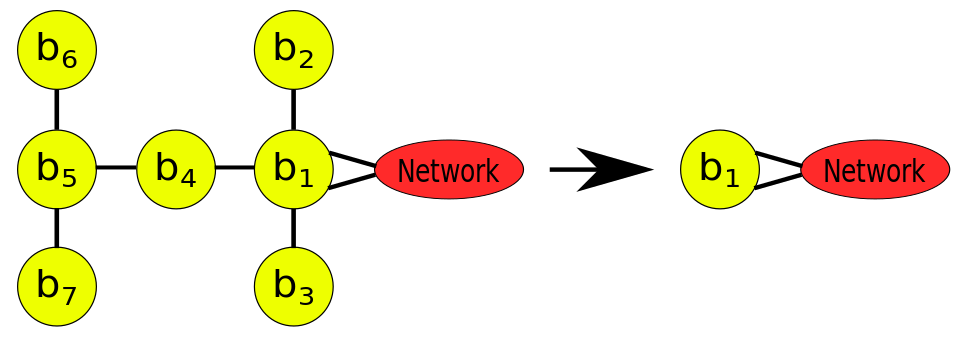}
\caption{Recursively collapsing nodes of degree one reduces the node set $\{b_2,b_3,b_4,b_5,b_6,b_7\}$ into the terminal super node $b_1$, where $\deg(b_1) \geq 2$.}
\end{figure}

Confining our analysis to the connected network $\mathbf{N}$, we collapse all trees in the network to their root nodes. This reduction is performed unambiguously by recursively mapping each node of degree one into the node with which it shares a line. The recursive step is performed until all nodes in $\mathbf{N}$ are of degree two or greater. Our method is described in \cref{alg:degree1}. To post process, and refine the graph structure, our design allows one to invert the collapse of any subset of a tree in the network; we use hashable maps for ease of implementation in this reconstruction. To each terminal super node, in which we cluster a tree, we associate a sequence of lists and arrays representing the recursive reduction procedure. This implementation is described comprehensively in \cref{sec:data}. We define the following notation.

\begin{definition}
The data structure $\mathbf{d1H}$ is a hashable map, $\{``field" :``data"\}$, where ``data'' is an ordered list. Any subset of nodes $\mathbf{t} \subset \mathbf{N}$ is defined as a \textbf{tree} if it is collapsed to a node under \cref{alg:degree1}. The mapping which collapses a tree, or a collection of trees, to the root node $b_1$ is associated to the field $t\_b_1$ in $\mathbf{d1H}$.
\end{definition}

\begin{algorithm}
\caption{Degree one reduction}
\label{alg:degree1}
\begin{algorithmic}
\STATE{Define: $\mathbf{d1N} \coloneqq \mathbf{N}$, $\mathbf{d1E} \coloneqq \mathbf{E}$, $\mathbf{d1H}\coloneqq$ empty hashable map.}
\WHILE{$\exists$ $b_1 \in \mathbf{d1N}$ with $\mathbf{deg}(b_1) < 2$,}
 \STATE{Remove $b_1$ from $\mathbf{d1N}$ and line $\{b_1,b_2\}$ from $\mathbf{d1E}$.}
 \IF{$t\_b_1 \in \mathbf{d1H}$,}
  \STATE{Append $b_2$ to each array stored in list $t\_b_1\in \mathbf{d1H}$.}
    \STATE{Append all arrays in list $t\_b_1\in \mathbf{d1H}$ to list $t\_b_2 \in \mathbf{d1H}$.}
      \STATE{Remove $t\_b_1$ from $\mathbf{d1H}$.}
    \ELSE
      \STATE{Write array $[b_1,b_2]$ to list $t\_b_2\in \mathbf{d1H}$.}
  \ENDIF
\ENDWHILE
\RETURN $\mathbf{d1N}, \mathbf{d1E}, \mathbf{d1H}$
\end{algorithmic}
\end{algorithm}

In each loop of \cref{alg:degree1} we collapse the degree one node, $b_1$, into the connected node $b_2$. The \textbf{if} statement requires that whenever a list of collapsed trees is associated to $t\_b_1 \in \mathbf{d1H}$, we append all associated arrays to the list $t\_b_2$, and $b_2$ is appended to each array denoting the root node. \cref{alg:degree1} reduces the test network to 32,891 nodes and 43,568 edges. The histogram of tree lengths and the distribution of the degrees of the nodes in $\mathbf{d1N}$ are given in the \cref{fig:d1_d2_dist}. The mean degree of nodes in $\mathbf{d1N}$ is 2.65, with a standard deviation 1.42 and maximal degree of 38. Tree lengths are calculated as the number of nodes aggregated into the super node, \emph{including} the root node itself. The total number of trees collapsed in $\mathbf{d1H}$ is 9,528 with a mean tree length 3.12 nodes, a standard deviation 2.41 and a max tree length of 36 nodes.  The graph density is approximately $8.054\times 10^{-5}$

The net power-flow, after mapping a tree to its root can be preserved in an intuitive way: the net power-flow of the entire tree can be parametrized through the terminal super node.  We will introduce the following lemma to demonstrate that this intuitive graph topological reduction is consistent with the procedure of the iterative Kron reduction.
\begin{lemma}\label{lemma:deg1step}
Let the loopy Laplacian $\bQ$ define an arbitrary, connected network of $n$ nodes satisfying \cref{hyp:inductive} and \cref{hyp:lossless}.  Without loss of generality, suppose $b_n$ is of degree one and is connected to node $b_{n-1}$.  Then under the Kron reduction with reference nodes defined by $\alpha \triangleq \{1,\cdots, n-1\}$, the entries of the reduced loopy Laplacian $\bQ^{\rm red}$, the reduced currents $\bC^{\rm red}$ and the reduced power injections $\bS^{\rm red}$ agree with the original $\bQ$, $\bC$ and $\bS$ in all entries except for those corresponding to the node $b_{n-1}$.  In particular, contracting the node $b_n$ into $b_{n-1}$ is realized by the Kron reduction with $\alpha $ as reference nodes.
\end{lemma}
\begin{proof}
Let $\be_{n-1}\in\mathbb{R}^{n-1}$ be the vector comprised of zeros, except for the value $1$ in the entry $n-1$.  We decompose the loopy Laplacian, $\bQ$, such that it is given by
\begin{align}
\bQ \triangleq 
\begin{pmatrix}
\bQ_{[\alpha,\alpha]} & Q_{n-1,n}\be_{n-1} \\
Q_{n-1,n} \be_{n-1}^\T & Q_{n,n}
\end{pmatrix}.
\end{align}
From equation \cref{eq:kronll}, the Kron reduced loopy Laplacian is given by
\begin{align}
\bQ^{\rm red} &= \bQ_{[\alpha,\alpha]} - \frac{Q_{n-1,n}^2}{Q_{n,n}}
          \begin{pmatrix}
          0 & \cdots & 0 \\
          \vdots & \vdots & \vdots\\
          0 & \cdots & 1
          \end{pmatrix}
\end{align} 
which implies that the Kron reduction contracts $b_n$ into $b_{n-1}$.  Indeed, all elements of $\bQ_{[\alpha, \alpha]}$ remain unaffected, \emph{except} entry $Q_{n-1, n-1}$, which on the other hand is adjusted by the factor of $-\frac{Q_{n-1,n}^2}{Q_{n,n}}$ to find the power-flow equivalent nodal admittance at the root node $b_{n-1}$.  It is also easy to see that the only non-zero entry of $\bQ^{\rm ac}$ is in its entry $n-1$ by construction.
\end{proof}
\begin{cor}\label{cor:deg1kron}
Recursively collapsing degree one nodes, as in \cref{alg:degree1}, is compatible with the iterative Kron reduction, with reference nodes defined by $\done\bN$.  The reduced network, therefore, satisfies \cref{hyp:inductive} and \cref{hyp:lossless}.
\end{cor}
\begin{proof}
\cref{lemma:deg1step} guarantees that collapsing a single node, as in \cref{alg:degree1} is compatible with the Kron reduction.  \cref{theorem:iteration} shows that this procedure can be performed recursively, and equivalently, to the Kron reduction produced with $\done\bN$ as reference nodes.
\end{proof}

As a consequence of \cref{cor:deg1kron}, the reduced current balance and power-flow equations for the network defined by $\done\bN, \done\bE$ can be computed via equation \cref{eq:currentbalance} and \cref{eq:powerflow}, with respect to a single iteration of the Kron reduction with the nodes $\done\bN$ chosen as a reference.  However, by storing the sequential mappings in $\done\bH$, various levels of resolution can be introduced to a reduced order model by: (i) selecting the nodes in $\done\bN$ as reference nodes, and (ii) additionally selecting trees or sub-trees as reference nodes for a Kron reduction of $\bN$.  Reducing radial networks was performed in a similar fashion by \cite{huo2018power} as a preprocessing step to their $k-$nearest neighbors clustering approach.  We add to this discussion now with the rigorous proof of the compatibility of this topological reduction with the iterative Kron reduction. 

The reductions to the test network via \cref{alg:degree1} are significant, however, in an on-line reduction we may expect a further collapse yet. In our study, the degrees of nodes in the distribution network are defined by the physical lines connecting nodes, irrespective of the operational disconnecting. In practice, however, the operational switching for real power delivery further sparsifies the network and forms additional tree structures that would be collapsed under \cref{alg:degree1}. The operational structure typically changes in response to system faults and outages which may occur a few times a day \cite{deka2017structure}. Therefore, on-line graph reduction faces the additional challenge of efficiently learning the operational topology, based on incomplete information, and constructing a reduced model within the window of the current configuration.

\subsection{Degree two reductions}
Geographically distant sub-networks that have significant generation or load resources are linked for robustness of power delivery. In case of line failures within one area, the interconnected sub-networks can be configured to balance loads and generation around the failure. In a setting where the intermediate area between these sub-networks has low generation or load, the connection between them is comprised of long range transmission lines, as seen in grids in the USA, China and others \cite{hvdc}. These transmission lines are topologically modeled as string-like line sub-graphs of degree two nodes. Often, the intermediate nodes in the string lack significant generation or load and have negligible impact to network dynamics. The simple dynamical transmission structure of these degree two nodes motivates an intuitive model of the power-flow: we replace all nodes in the interior of the string with a ``meta-edge'' and parameterize the net power-flow with line characteristics. \cref{fig:edgemap} visualizes this reduction in which we approximate the string of degree two nodes, $b_1$ through $b_5$, with a single line connecting $b_1$ and $b_5$.

\begin{figure}[ht]\label{fig:edgemap}
\center
\includegraphics[width=.5\linewidth]{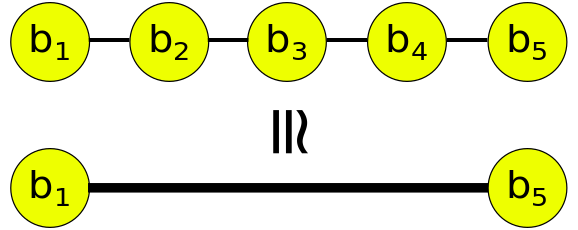}
\caption{Strings of degree two nodes are mapped to a ``meta-edge''.}
\end{figure}

In the following lemma, we demonstrate the simple case of a degree two topological reduction, mapping a degree two node to an edge, is consistent with the iterative Kron reduction.
\begin{lemma}\label{lemma:simpledegree2}
Let the loopy Laplacian $\bQ$ define an arbitrary, connected network of $n$ nodes satisfying \cref{hyp:inductive} and \cref{hyp:lossless}.  Without loss of generality, suppose $b_n$ is of degree two and is connected to nodes $b_{n-2}$ and $b_{n-1}$.  Then under the Kron reduction with reference nodes defined by $\alpha \triangleq \{1,\cdots, n-1\}$, the entries of the reduced loopy Laplacian $\bQ^{\rm red}$, the reduced currents $\bC^{\rm red}$ and the reduced power injections $\bS^{\rm red}$ agree with the original $\bQ$, $\bC$ and $\bS$ in all entries except for those corresponding to nodes $b_{n-1}$ and $b_{n-2}$.  In particular, mapping the node $b_n$ to the edge $\{b_{n-1},b_{n-2}\}$ is realized by the Kron reduction with reference nodes defined by $\alpha$.
\end{lemma} 
\begin{proof}
Let $\be_{n-2},\be_{n-1}\in \mathbb{R}^{n-1}$ be the vectors composed of zeros except for a one in position $n-2$ and $n-1$ respectively.  We decompose the loopy Laplacian $\bQ$ as,
\begin{align}
\bQ = 
\begin{pmatrix}
\bQ_{[\alpha, \alpha]} & Q_{n-2,n}\be_{n-2} + Q_{n-1,n}\be_{n-1} \\
Q_{n-2,n}\be_{n-2}^\T +Q_{n-1,n}\be_{n-1}^\T & Q_{n,n} 
\end{pmatrix}
\end{align}
The Kron reduced loopy Laplacian is given by 
\begin{align}
 \bQ^{\rm red} = \bQ_{[\alpha,\alpha]} - \frac{1}{Q_{n,n}}
          \begin{pmatrix}
          0 & \cdots & 0 & 0& 0 \\
          \vdots & \ddots &\vdots & \vdots & \vdots\\
          0 & \cdots & 0& 0 & 0\\
          0 & \cdots &0 &  Q_{n-2,n}^2 & Q_{n-2,n} Q_{n-1,n} \\
          0 & \cdots & 0 & Q_{n-2,n} Q_{n-1,n} & Q_{n-2,n}^2 
          \end{pmatrix},
\end{align} 
such that the admittance of the line $\{b_{n-2},b_{n-1}\}$ is updated in the Kron reduced network, where 
\begin{align}
Q^{\rm red}_{n-2,n-1} = Q^{\rm red}_{n-1,n-2} = Q_{n-2,n-1} -\frac{Q_{n-2,n} Q_{n-1,n}}{ Q_{n,n}}.
\end{align}
Similarly, the nodal self admittances of $b_{n-2},b_{n-1}$ are updated such that
\begin{align}
Q^{\rm red}_{n-2,n-2} &= Q_{n-2,n-2} - \frac{ Q_{n-1,n}^2}{ Q_{n,n}}\\
Q^{\rm red}_{n-2,n-2} &= Q_{n-1,n-1} -\frac{ Q_{n-2,n}^2 }{ Q_{n,n}},
\end{align}
while leaving all other nodes unaffected.  It is also easy to see that the only non-zero entries of $\bQ^{\rm ac}$ are in entries $n-2$ and $ n-1$ by construction.
\end{proof}

Recursively removing degree one nodes from $\mathbf{N}$ as described \cref{section:degreeone} produces the network $\mathbf{d1N}$, $\mathbf{d1E}$ comprised of nodes degree two or greater. Our subsequent topological reduction proceeds to remove all nodes of degree two by recursively replacing degree two nodes with edges, \emph{if the edge does not already appear in $\mathbf{d1E}$}.  However, our topology-based procedure of: (i) mapping the node $b_n$ to the edge $\{b_{n-2},b_{n-1}\}$, (ii) prohibiting multiple edges between nodes, has the additional effect of reducing other tree-like configurations. These structures are discovered when the procedure results in a degree one node in $\mathbf{d1N}$.  

\begin{definition}\label{def:sparsetriangle}
Let $\{b_1, b_2, b_3\} \subset \mathbf{d1N}$ and $\{b_1,b_2\},\{b_1,b_3\},\{b_2,b_3\}\in \mathbf{d1E}$ such that $\mathbf{deg}(b_1)=\mathbf{deg}(b_2) =2$, while $\mathbf{deg}(b_3)\geq 3$. The set $\{b_1, b_2, b_3\} \subset \mathbf{d1N}$ is defined to be a \textbf{sparsely connected triangle}.
\end{definition}

Let $\{b_1, b_2, b_3\}$ be a sparsely connected triangle as in the left hand side of \cref{fig:gentree}. Removing $b_1$, and the edges $\{b_1,b_2\}$ and $\{b_1,b_3\}$, lowers the degree of $b_2$ to one. In \cref{lemma:gentree} we demonstrate that a degree one node is produced by our reduction procedure if and only if a sparsely connected triangle is reduced via removing a degree two node contained in the triangle.

\begin{figure}[ht]\label{fig:gentree}
\center
\includegraphics[width=.7\linewidth]{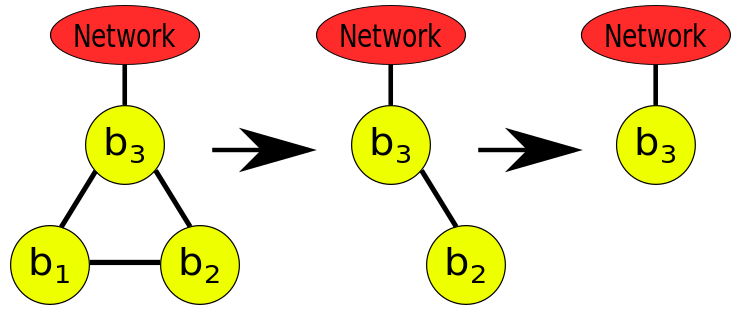}
\caption{\cref{alg:nodetoedge} removes the node $b_1$ from the sparsely connected triangle. \cref{alg:collapselasso} maps the nodes $b_1$ and $b_2$ to the root node $b_3$.  Note that the sparsely connected triangle ${b_1,b_2, b_3}$ will further collapsed into the network by \cref{alg:collapselasso} if $\deg(b_3)=3$ at the beginning of the reduction.}
\end{figure}

\begin{lemma}\label{lemma:gentree}
Let $b_1\in \mathbf{d1N}$ be a node of degree two with edges $\{b_1,b_2\}$ and $\{b_1,b_3\}$. A degree one node is produced by replacing $b_1$ with the edge $\{b_2,b_3\}$, prohibiting double lines, if and only if $\{b_1,b_2,b_3\}$ is a sparsely connected triangle.
\end{lemma}
\begin{proof}
By construction, every node in $\mathbf{d1N}$ is of degree two or greater. The nodes $b_2$ and $b_3$ each lose an edge in the reduction, $\{b_1,b_2\}$ and $\{b_1,b_3\}$ respectively. However, if $\{b_2,b_3\}$ is not an edge we will add this edge to the network. Therefore $b_2$ and $b_3$
remain the same degree if and only if $\{b_2,b_3\} \notin \mathbf{d1E}$. Suppose removal of the edges $\{b_1,b_2\}$ and $\{b_1,b_3\}$ has produced a degree one
node. We conclude $\mathbf{deg}(b_2)=2$ or $\mathbf{deg}(b_3)=2$, and $\{b_2,b_3\}\in \mathbf{d1E}$. Without loss of generality suppose $\mathbf{deg}(b_2)=2$. The edges connecting $b_2$ are therefore $\{b_1,b_2\}$ and $\{b_2,b_3\}$. The network $\mathbf{d1N}$ has a single connected component so we conclude that $\mathbf{deg}(b_3)\geq 3$. Indeed, this node must connect the triangular to the rest of the network. The converse statement is obvious from the above discussion and \cref{fig:gentree}. By subsequently performing a recursive collapse of degree one nodes, we may redefine $\mathbf{d1N}$ to consist of nodes at least degree two.
\end{proof}

Many structures reduce to a sparsely connected triangle by recursively replacing nodes with edges.  The configurations which reduce to sparsely connected triangles includes but is not limited to 
\begin{enumerate}[(i)]
  \item any simple polygon of nodes $\mathbf{P} \subset \mathbf{d1N}$ for which every node but one in $\mathbf{P}$ is of degree two;
  \item many triangular meshes, which as in \cref{fig:meshy_triangles}, are connected to the outside network through a single node;
  \item various combinations of the above.  
\end{enumerate}  
\begin{figure}[ht]\label{fig:meshy_triangles}
\center
\includegraphics[width=.3\linewidth]{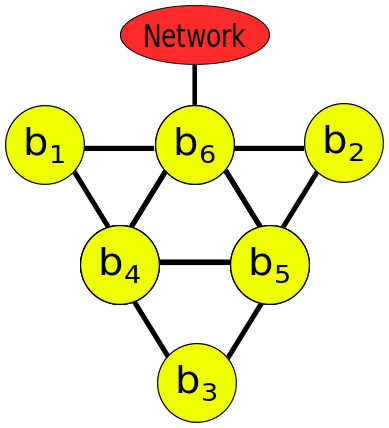}
\caption{Replacing nodes $b_1,b_2$ and $b_3$ with edges, while prohibiting multiple edges reduces the triangular mesh to a sparsely connected triangle.}
\end{figure}

An exhaustive characterization of the possible configurations which reduce to sparsely connected triangles is difficult to define, and goes beyond the scope of this work.  However, the above examples are useful for intuitively demonstrating the types of sub-networks which are collapsed to sparsely connected triangles.  Specifically, these are formed by loopy meshes, which might be densely connected internally, but are connected to other sub-networks of nodes \emph{through a single terminal bus}.  These formations are typical of distribution structures, and coherent sets of nodes that are weakly coupled to the rest of the network, and this justifies modeling these configurations as generalized trees.  

Actually, any configuration of nodes that can be reduced to a sparsely connected triangle is collapsed entirely to a terminal node.  Specifically, the sparsely connected triangle is broken by our routine as in \cref{lemma:gentree}, and the remaining nodes are mapped into a terminal root node by recursive degree one reduction.  In the following lemma we demonstrate that mapping the sparsely connected triangle, described in \cref{def:sparsetriangle}, into the root node $b_3$ is consistent with the iterative Kron reduction.

\begin{lemma}\label{lemma:sparsetrikron}
Let the loopy Laplacian $\bQ$ define an arbitrary, connected network of $n$ nodes satisfying \cref{hyp:inductive} and \cref{hyp:inductive}.  Without loss of generality, suppose that the set of nodes $\{b_{n-2},b_{n-1}, b_n\}$ forms a sparsely connected triangle.  Then under the Kron reduction with reference nodes defined by $\alpha \triangleq \{1,\cdots, n-2\}$, the entries of the reduced loopy Laplacian $\bQ^{\rm red}$, the reduced currents $\bC^{\rm red}$ and the reduced power injections $\bS^{\rm red}$ agree with the original $\bQ$, $\bC$ and $\bS$ in all entries except for those corresponding to node $b_{n-2}$.  In particular, mapping the nodes $\{b_{n-2},b_{n-1},b_n\}$ into the root node $b_{n-2}$ is realized by the Kron reduction  with reference nodes defined by $\alpha$.
\end{lemma}
\begin{proof}
Let $\be_{n-2}\in \mathbb{R}^{n-2}$ be the vector composed of zeros except for a one in position $n-2$.  We decompose the loopy Laplacian $\bQ$ as,
\begin{align}
\bQ = 
\begin{pmatrix}
\bQ_{[\alpha, \alpha]} & Q_{n-2,n-1}\be_{n-2}  &Q_{n-2,n}\be_{n-2}   \\
Q_{n-2,n-1}\be_{n-2}^\T & Q_{n-1,n-1} & Q_{n-1,n}\\
Q_{n-2,n}\be_{n-2}^\T  & Q_{n-1,n} &  Q_{n,n}
\end{pmatrix}.
\end{align}
Thus,
\begin{align}
\bQ_{(\alpha, \alpha)}^{-1}
=\frac{1}{Q_{n-1,n-1}Q_{n,n} -Q_{n-1,n}^2} 
\begin{pmatrix}
 Q_{n,n} & -Q_{n-1,n} \\
 -Q_{n-1,n}& Q_{n-1,n-1}
\end{pmatrix}
\end{align}
which implies
\begin{align}
\bQ_{[\alpha,\alpha)} \bQ_{(\alpha,\alpha)}^{-1} \bQ_{(\alpha,\alpha]} &=
  \begin{pmatrix}
  \0_{n-2 \times n-2} &  \0_{n-2\times 1} \\
  \0_{1 \times n-2} & q
  \end{pmatrix},
\end{align} 
where $q$ is a scalar, computed directly as
\begin{align}\begin{split}
q =& 
\frac{Q_{n-2,n-1}\(Q_{n-2,n-1} Q_{n,n} -Q_{n-1,n}Q_{n-2,n}\)}{Q_{n-1,n-1}Q_{n,n} -Q_{n-1,n}^2}  \\
&+\frac{Q_{n-2,n}\(Q_{n-1,n} Q_{n-1,n-1} - Q_{n-2,n-1} Q_{n-1,n}\)}{Q_{n-1,n-1}Q_{n,n} -Q_{n-1,n}^2}.
\end{split}\label{eq:qterm}
\end{align}

The Kron reduced loopy Laplacian is given by
\begin{align}
\bQ^{\rm red} &=
 \bQ_{[\alpha,\alpha]} - \begin{pmatrix}
  \0_{n-2 \times n-2} &  \0_{n-2\times 1} \\
  \0_{1 \times n-2} & q
  \end{pmatrix},
\end{align}
such that the sparsely connected triangle is collapsed into the node $b_{n-2}$ while leaving all other nodes unaffected.  The self admittance for $b_{n-2}$ is updated in the reduced model via the term $-q$ defined in equation \cref{eq:qterm}.  Finally, it is easy to verify that the only non-zero entries of $\bQ^{\rm ac}\in\mathbb{C}^{n-2 \times 2}$ are those in row $n-2$.
\end{proof}

Our analysis of the basic degree two reduction, and the reduction of sparsely connected triangles, leads to \cref{alg:nodetoedge} and \cref{alg:collapselasso}.  We describe the data structures used in these routines in \cref{sec:data} and define the following notation.
\begin{definition}
The data structure $\mathbf{d2H}$ is a hashable map $\{``field" :``data"\}$, where ``data'' is an ordered list. The mapping which takes the node $b_1$ to the edge $\{b_2, b_3\}$ is associated to the field $e\_b_2\_b_3$, where we assume $b_2 < b_3$. We define any subset $\mathbf{gt} \subset \mathbf{d1N}$ to be a \textbf{generalized tree} if it is collapsed to a root node under \cref{alg:nodetoedge} and \cref{alg:collapselasso}. The mapping which collapses a generalized tree to the terminal node $b_1$ is associated to the field $t\_b_1\in \mathbf{d2H}$.
\end{definition}

\begin{algorithm}
\caption{Degree two reduction}
\label{alg:nodetoedge}
\begin{algorithmic}
\STATE{Define: $\mathbf{d2N} \coloneqq\mathbf{d1N}$, $\mathbf{d2E} \coloneqq\mathbf{d1E}$ and $\mathbf{d2H}\coloneqq\mathbf{d1H}$.}
\WHILE{$\exists$ $b_1 \in \mathbf{d2N}$ with $\mathbf{deg}(b_1) < 3$,}
 \STATE{Let $b_1$ be connected to $b_2$ and $b_3$ S.T. $b_2 < b_3$.}
 \STATE{Remove $b_1$ from $\mathbf{d2N}$ and lines $\{b_1, b_2\}$,$\{b_1, b_3\}$ from $\mathbf{d2E}$.}
 \IF{$\{b_2, b_3\} \notin \mathbf{d2E}$,}
  \STATE{Write $\{b_2, b_3\} \in \mathbf{d2E}$.}
 \ENDIF
  \IF{$\exists$ $b_j$ S.T. $e\_b_1\_b_j$ or $e\_b_j\_b_1 \in \mathbf{d2H}$,}
    \STATE{Append list entries in $e\_b_1\_b_j$ or $e\_b_j\_b_1$ to list $e\_b_2\_b_3 \in \mathbf{d2H}$.}
    \STATE{Remove $e\_b_1\_b_j$ or $e\_b_j\_b_1$ from $\mathbf{d2H}$}
  \ENDIF
  \STATE{Append $[b_2,b_1,b_3]$ to list $e\_b_2\_b_3 \in \mathbf{d2H}$.}
  \IF{$t\_b_1\in \mathbf{d2H}$,}
    \STATE{Append the hashable map $\{t\_b_1 : \mathbf{d2H}(t\_b_1)\}$ to list $e\_b_2\_b_3\in \mathbf{d2H}$.}
    \STATE{Remove $t\_b_1$ from $\mathbf{d2H}$.}
  \ENDIF
  \STATE{Pass $\mathbf{d2N}$, $\mathbf{d2E}$ and $\mathbf{d2H}$ to \cref{alg:collapselasso}.}
\ENDWHILE
\RETURN $\mathbf{d2N},\mathbf{d2E}, \mathbf{d2H}$
\end{algorithmic}
\end{algorithm}

\begin{algorithm}
\caption{Reduce sparsely connected triangle}
\label{alg:collapselasso}
\begin{algorithmic}
\IF{$\exists$ $a_1 \in \mathbf{d2N}$ with $\mathbf{deg}(a_1) < 2$,}
\WHILE{$\exists$ $a_1 \in \mathbf{d2N}$ with $\mathbf{deg}(a_1) < 2$,}
  \STATE{Remove $a_1$ from $\mathbf{d2N}$ and line $\{a_1,a_2\}$ from $\mathbf{d2E}$.}
    \IF{$t\_a_1 \in \mathbf{d2H}$,}
      \STATE{Append $a_2$ to each array stored in list $t\_a_1\in \mathbf{d2H}$.}
      \STATE{Append all arrays in list $t\_a_1\in \mathbf{d2H}$ to list $t\_a_2 \in \mathbf{d2H}$.}
      \STATE{Remove $t\_a_1$ from $\mathbf{d2H}$.}
    \ELSE
      \STATE{Write array $[a_1,a_2]$ to list $t\_a_2\in \mathbf{d2H}$.}
 \ENDIF
\ENDWHILE
\STATE{Prepend hash table $\{e\_b_2\_b_3: \mathbf{d2H}(e\_b_2\_b_3)\}$ to list $t\_a_2\in \mathbf{d2H}$.}
\STATE{Remove $e\_b_2\_b_3$ from $\mathbf{d2H}$.}
\ENDIF
\RETURN $\mathbf{d2N},\mathbf{d2E}, \mathbf{d2H}$
\end{algorithmic}
\end{algorithm}

\cref{alg:nodetoedge} maps nodes to edges and tracks these reductions sequentially in the hashable map $\mathbf{d2H}$. Whenever $b_1$ is mapped to the edge $\{b_2,b_3\}$, if $\{b_1, b_2\}$ or $\{b_1, b_3\} \in \mathbf{d2H}$, we write all preceding mappings to the list $e\_b_2\_b_3$ when $\{b_1, b_2\}$ and $\{b_1, b_3\}$ are removed. We enforce a similar condition whenever a generalized tree is associated to the node $b_1$.  If $t\_b_1\in \mathbf{d2H}$, these maps are appended, as a hashable map, to the list $e\_b_2\_b_3$. The subroutine, \cref{alg:collapselasso}, is a modification of \cref{alg:degree1} which tracks the collapse of sparsely connected triangles. Knowing that a degree one node is produced under \cref{alg:nodetoedge} if and only if the routine breaks a sparsely connected triangle, \cref{alg:collapselasso} stores the list $e\_b_2\_b_3$ under in the root of the generalized tree subsequently collapsed. The root is defined by the final iteration of the degree one reduction. The design and inversion of these data structures is described in detail in \cref{sec:data}.

\begin{cor}\label{cor:deg2kron}
Recursively collapsing degree degree two nodes and sparsely connected triangles, as in \cref{alg:nodetoedge} and \cref{alg:collapselasso}, is compatible with the iterative Kron reduction, with reference nodes defined by $\dtwo\bN$. The reduced network, therefore, satisfies \cref{hyp:inductive} and \cref{hyp:lossless}.
\end{cor}
\begin{proof}
\cref{lemma:simpledegree2} and \cref{lemma:sparsetrikron} guarantees that a single iteration of \cref{alg:nodetoedge} and \cref{alg:collapselasso} is compatible with the Kron reduction. \cref{theorem:iteration} shows that this procedure can be performed recursively, and equivalently, to the Kron reduction produced with $\dtwo\bN$ as reference nodes.
\end{proof}

\begin{cor}\label{cor:deg2graphpath}
Let $b_1, b_2\in \dtwo\bN$, the network reduced via \cref{alg:degree1}, \cref{alg:nodetoedge} and \cref{alg:collapselasso}.  There exists an edge between $b_1$ and $b_2$ if and only if there exists a path from $b_1$ to $b_2$ in $\{\bN, \bE\}$ such that all interior nodes in the path belong to $\bN \setminus \dtwo \bN$. That is, \cref{alg:degree1}, \cref{alg:nodetoedge} and \cref{alg:collapselasso} preserve graph paths. 
\end{cor}
\begin{proof}
\cref{cor:deg1kron} and \cref{cor:deg2kron} demonstrate that the reduced network $\dtwo\bN$ defined by \cref{alg:degree1}, \cref{alg:nodetoedge} and \cref{alg:collapselasso}, is equivalent to the Kron reduction of $\bN$ with $\dtwo\bN$ chosen as reference nodes.  But Theorem III.4 of \cite{dorfler2013kron} demonstrates that the Kron reduction preserves graph paths --- the above statement is a simple corollary.
\end{proof}

\cref{alg:nodetoedge} and \cref{alg:collapselasso} reduce the sets $\mathbf{d1N},\mathbf{d1E}$ to the sets $\mathbf{d2N},\mathbf{d2E}$ with 9716 nodes and 18,700 edges. \cref{fig:d1_d2_dist} summarizes this reduction with the histogram of the number of nodes per reduction in $\mathbf{d2H}$ and the distribution of the degrees of nodes in $\mathbf{d2N}$. The mean degree of nodes in $\mathbf{d2N}$ is 3.85, with a standard deviation 1.76 and a maximal node degree of 38.  The graph density is approximately $3.9622\times 10^{-4}$.  The total number of collapsed edges in $\mathbf{d2H}$ is 9,696, with the mean number of nodes per edge is 3.88, standard deviation 4.24 and max nodes per edge 94. The total number of generalized trees collapsed in $\mathbf{d2H}$ is 2,579 with a mean of 4.38 nodes per generalize tree, standard deviation of 3.65 and max nodes per generalized tree 56. We note, generalized trees which have been mapped to edges are considered only as nodes within the meta-edge of their final reduction. Likewise, we do not distinguish meta-edges which have been collapsed into generalized trees from the root super node where their reduction terminates.

\begin{figure}\label{fig:d1_d2_dist}
\includegraphics[width=\linewidth]{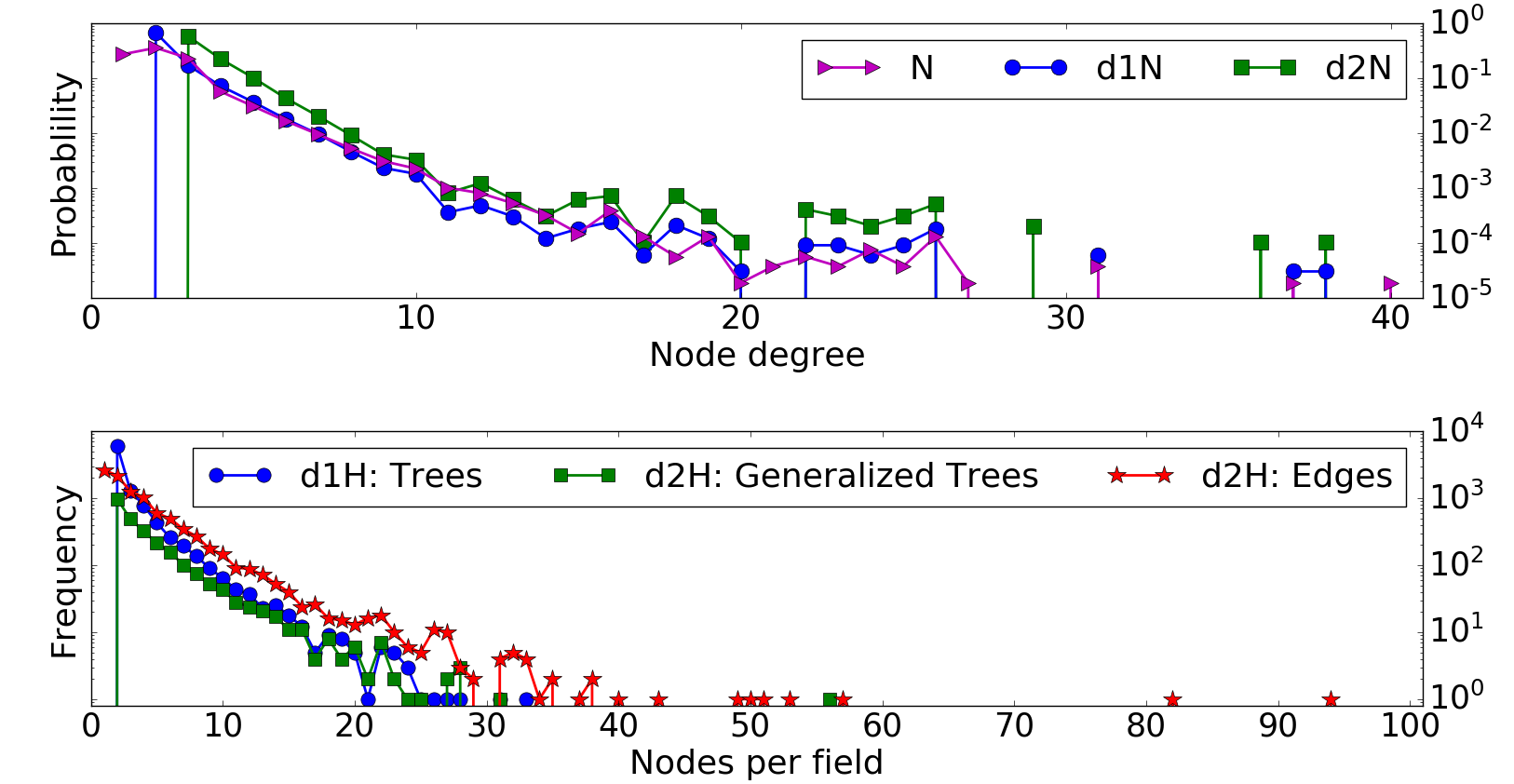}
\center
\caption{\textbf{Upper:} Distribution of the degrees of nodes in $\mathbf{N}$, $\mathbf{d1N}$ and $\mathbf{d2N}$.  The distance between the distribution for $\done \bN$ and $\bN$ is approximately $0.3588$.  The distance between the distribution for $\done \bN$ and $\bN$ is approximately $1.4477$. \textbf{Lower} Histogram of nodes per tree in $\mathbf{d1H}$, generalized tree in $\mathbf{d2H}$ and meta-edge in $\mathbf{d2H}$.}
\end{figure}

We note that the degree one and degree two reductions generally fail to preserve the shape of the degree distribution for the full network.  This is to be expected as the degree distribution for the full network is strongly peaked around nodes with degree one to three, while the above algorithms ensure that the remaining network has no nodes of degree less that three.  In particular, we compute the first Wasserstein distance between the degree one reduced, degree two reduced, and unreduced networks with an L1 ground distance on the discrete bins corresponding to the nodal degrees.  The degree distribution for $\done \bN$ differs from that of $\bN$ by approximately $0.3588$, while the degree distribution for $\dtwo \bN$ differs from that of $\bN$ by approximately $1.4477$.  The distance of $1.4477$ between the degree distributions for $\dtwo \bN$ and $\bN$ can be equated in terms of the work of translating the degree distribution for $\bN$ by almost $1.5$ bins to the right.

  However, the bias introduced in the degree distribution in the reduced network is justified given the design features of the nodes which are reduced, and their physical role in multiscale power grids.  Typically, the nodes reduced are formed by coherent sub-networks which are weakly coupled to the remaining network.  Despite the distortion of the degree distribution in the reduced network, the topological approach used in selecting reference nodes for the Kron reduction has the advantage of preserving the sparsity of the original network.  Particularly, the above approach has the benefit of maintaining the weak and strong coupling between coherent sub-networks of nodes, which is not generally guaranteed with an arbitrary selection of reference nodes.

It is still possible, however, that an unintended mixing of loads, generation, distribution and transmission structures will occur with the above topological approach to nodal aggregation.  For this reason, a user may invert reductions produced by \cref{alg:degree1}, \cref{alg:nodetoedge} and \cref{alg:collapselasso} post-facto to refine the resolution on a particular aggregation, and include factors such as line admittances and nodal voltage in selecting reference nodes.  This will be a central concern as we introduce additional steps to produce further reductions to the network, which will aggregate higher degree coherent sub-networks.

\subsection{Triangular reductions}

Our work in the previous sections shows that \cref{alg:degree1}, \cref{alg:nodetoedge} and \cref{alg:collapselasso} produce a network, $\mathbf{d2N}$, $\mathbf{d2E}$, that is of a scale which permits qualitative analysis.  The compatibility of the topological reduction with the iterative Kron reduction implies that under the assumption of a lossless, inductive, steady state network, the usual analysis with dynamic swing equations may be applied for analysis of optimal power-flow and control \cite{dorfler2013kron}.  The dynamic swing equation computations in transient stability studies require a reduced model, even when performed off-line, from the original 53,155 bus network.  Here, a network with less than 10,000 nodes is feasible for the off-line simulation.  The size of the network may remain a bottleneck, however, for on-line computations.  On-line applications vary, some of which require optimal power-flow or dynamic and nonlinear simulation.   Likewise, further reductions to the network may be necessary to make on-line parameter estimation in an operational window feasible, where inertial and damping coefficients may only remain constant on the order of minutes \cite{lokhov2018online}. 

\begin{figure}[ht]\label{fig:trred}
\center
\includegraphics[width=.8\linewidth]{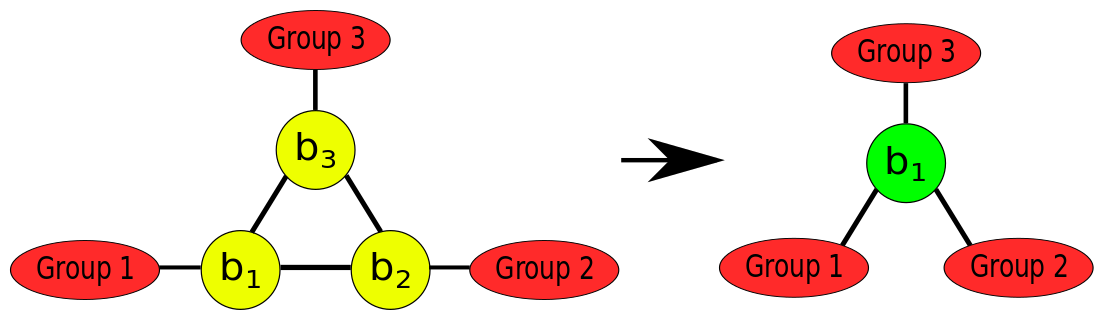}
\caption{A ``pure'' triangular configuration.  The triangle is collapsed to $b_1$, which gains all lines from $b_2$ and $b_3$ avoiding self lines.}
\end{figure}

It is possible to collapse higher degree coherent structures, such as non-sparse triangular configurations, but there is greater subtlety. The degree one and degree two node reductions will produce an unambiguous model for net power-flow via a direct implementation of the Kron reduction on the full network with the reference nodes defined via \cref{alg:degree1}, \cref{alg:nodetoedge} and \cref{alg:collapselasso}.  We are motivated to perform a similar reduction to ``pure'' triangular reductions pictured in \cref{fig:trred} where there are three nodes, each of degree three and similar nominal voltage, forming a link between three large connected groups of nodes.  In this case we wish to collapse the three nodes $\{b_1,b_2,b_3\}$ on the left to a single super node of degree three on the right, such that the super node: (i) receives all lines from nodes $\{b_1,b_2,b_3\}$, excluding double and self lines; (ii) combines the currents of the three nodes.  In the transient stability regime, this reduction preserves the net power-flow through the triangle formed by nodes $\{b_1,b_2,b_3\}$ into the other sub-networks in \cref{fig:trred}.  Given that the three nodes $\{b_1,b_2,b_3\}$ are of similar nodal voltage, this may also approximate the mid-term stability dynamics.

While ``pure'' triangular configurations are easy to picture, they are rare and many other triangular configurations exist throughout the network.  Generically in \cref{fig:trred}, the three groups of nodes connected to the buses $\{b_1,b_2,b_3\}$ may be interconnected and the set of nodes in each group overlapping. Recursively collapsing generic triangular configurations to super nodes may generally produce multiple lines between nodes, raise and lower the degree of the super nodes (and the surrounding nodes), and produce non-unique final reductions.  For example in \cref{fig:tr_deg_loss}, when triangle formed by nodes $\{b_1,b_2,b_3\}$ is collapsed to a super node in the right side, the degree of node $b_4$ actually decreases, even though it was not directly included in the reduction. In this section, we will develop further reduction steps to produce a physically consistent network, but due to these subtleties, the methodology will become slightly more ad-hoc.

\begin{figure}[ht]\label{fig:tr_deg_loss}
\center
\includegraphics[width=.6\linewidth]{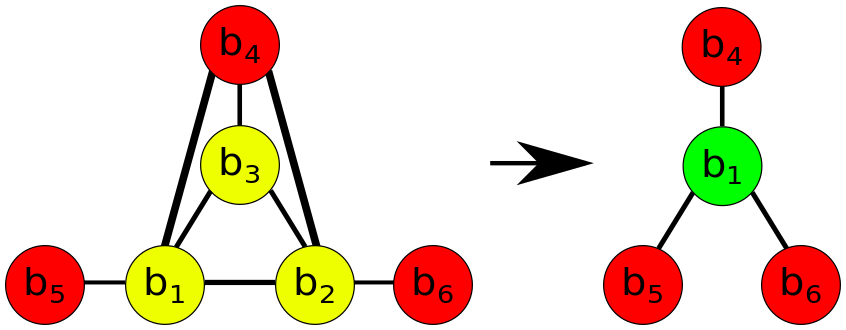}
\caption{A generic triangular configuration.  The triangle formed by nodes $\{b_1,b_2,b_3\}$ is collapsed to $b_1$, which gains all lines from $b_2$ and $b_3$, avoiding self lines.  In this case, the degree of the node $b_4$ will decrease by two, as the lines connecting $b_4$ to $b_1,b_2$ and $b_3$ are combined.}
\end{figure}

\begin{rmk}
Though we will define a reduced network via collapsing triangles that is power-flow equivalent in the same sense as the Kron reduction, this type of reduction is fundamentally different from the earlier algorithms and is not equivalent to any Kron reduced network.
\end{rmk}

Arbitrarily collapsing triangles in the network may strongly bias the distribution of node degrees, and possibly change the sparsity of the network. To prevent this bias, we permit the collapse of a triangle only if each node in the configuration does not exceed a specified degree. Recursively mapping triangles to nodes, the degrees of nodes within the reduced network will increase and decrease, so that we always refer to the degree of each node in the \emph{current iteration} of the algorithm.  The degree threshold introduces a tuning factor into our algorithm with which we balance the scale of the collapse with preserving the graph sparsity and degree distribution for nodes in $\mathbf{d2N}$. 

For any degree threshold $\mathbf{dThr}$, the maximal degree of a \emph{super node} produced by collapsing a triangular configuration is given by $3(\mathbf{dThr} - 2)$.  For example, assume that the set of nodes $\{b_1, b_2, b_3\}$ forms a triangle and each node has the maximum of $\mathbf{dThr}$ lines.  Let $i,j\in \{1,2,3\}$ and $k\in \{4, \cdots, n\}$.  A super node of degree $3(\mathbf{dThr} -2 )$ is produced collapsing $b_1,b_2,b_3$ if, when $b_i$ is connected to the node $b_k$, then $b_j$ is not connected to $b_k$ for each $j\neq i$.  In particular, each node $b_j$ contributes $\mathbf{dThr} -2$ distinct lines to the aggregated super node, after the lines that connect $b_1,b_2,$ and $b_3$ are removed.  We choose $\mathbf{dThr}=6,7$ and $8$, which produce a super node of degree at most $12,15$ and $18$ respectively.

A solely graph based reduction of triangles may also combine transmission and distribution nodes in a way which distorts the dynamics in the transient stability regime. For instance, if the ``pure'' triangle in \cref{fig:trred} is formed by two nodes, $b_1$ and $b_2$, of high nominal voltage while $b_3$ is of low nominal voltage, the super node produced from clustered triangle will confer stronger coupling between the three separate sub-networks (groups one, two, and three) than actually exists. To prevent non-coherent mixing of transmission and distribution sub-networks, we restrict our reductions only to the nodes in $\mathbf{d2N}$ which fall below an additional voltage threshold: we will permit a reduction to a triangle if every node in the configuration additionally falls below a specified nominal voltage. We choose voltage thresholds of $110, 138, 230, 345$ nominal KV (standard low, medium and high voltages for different transmission grid lines), and for reference, compare results without a voltage threshold.

Due to the earlier steps, the nodes in $\mathbf{d2N}$, and edges in $\mathbf{d2E}$, may represent multiple nodes due to reductions performed in \cref{alg:degree1}, \cref{alg:nodetoedge} and \cref{alg:collapselasso}. Our analysis leads to \cref{alg:triangular}, we introduce the following notation.
\begin{definition}
Let $\mathbf{vThr}$ be a specified voltage threshold. Define $\mathbf{nL}$ to be a list of nodes in $\mathbf{d2N}$ excluding any node(s)
\begin{itemize}
 \item $b_1$ such that $t\_b_1 \in\mathbf{d2H}$ contains a node of nominal voltage above $\mathbf{vThr}$,
 \item $b_1,b_2$ such that $e\_b_1\_b_2 \in \mathbf{d2H}$ contains a node of nominal voltage above $\mathbf{vThr}$
 \item or $b_1\in \mathbf{d2N}$ which has a nominal voltage above $\mathbf{vThr}$.
\end{itemize}
The data structure $\mathbf{triH}$ is a hashable map $\{``field":``data"\}$ where ``data'' is an ordered list. Entries of these lists are hashable maps of the form $\{ ``b_1\text{''} : \mathbf{lines}(b_1)\}$ where $\mathbf{lines}(b_1)$ is a list of lines associated to $b_1$ in $\mathbf{d2E}$.
\end{definition}

\begin{algorithm}
\caption{Greedy triangular reduction}
\label{alg:triangular}
\begin{algorithmic}
\STATE{Define: $\mathbf{triN}\coloneqq\mathbf{d2N}$, $\mathbf{triE}\coloneqq\mathbf{d2E}$, $\mathbf{triH}\coloneqq$ empty hashable map.}
\STATE{$\mathbf{nL}\coloneqq$ random permutation of $\mathbf{nL}$.}
\STATE{$\mathbf{dThr}\coloneqq$ degree threshold, $K\coloneqq 0$, $STOP\coloneqq \text{length}(\mathbf{nL})$.}
\WHILE{$K<STOP$,}
 \STATE{$K\coloneqq K+1$, $b_1 \coloneqq \mathbf{nL}(K)$.}
  \IF{$\mathbf{deg}(b_1) < \mathbf{dThr}$,}
    \WHILE{$\exists$ a triangular configuration defined by $\{b_1,b_2\},\{b_1,b_3\}, \{b_2,b_3\} \in \mathbf{triE}$ where $\mathbf{deg}(b_2),\mathbf{deg}(b_3)<\mathbf{dThr}$ and $b_2,b_3 \in \mathbf{nL}$,}
      \FOR{$b_i \in \{b_2,b_3\}$,}
        \STATE{Append $\{ ``b_i\text{''} : \mathbf{lines} (b_i) \}$ to $tri\_b_1 \in \mathbf{triH}$. }
        \STATE{Append all entries in $tri\_b_i$ to $tri\_b_1 \in \mathbf{triH}$.}
        \STATE{Remove $tri\_b_i$ from $\mathbf{triH}$.}
        \FOR{each $b_j$ such that $\{b_i,b_j\} \in \mathbf{triE}$,}
          \STATE{Write $\{b_1, b_j\}$ to $\mathbf{triE}$ excluding double and self lines.}
          \STATE{Remove $\{b_i,b_j\}$ from $\mathbf{triE}$.}
        \ENDFOR
       \STATE{Remove $b_i$ from $\mathbf{triN}$ and from $\mathbf{nL}$.}
      \ENDFOR
      \STATE{$K \coloneqq 0$, $\mathbf{nL}\coloneqq$ random permutation of $\mathbf{nL}$, $STOP\coloneqq\text{length}(\mathbf{nL})$.}
    \ENDWHILE
  \ENDIF
\ENDWHILE
\FOR{$tri\_b_1 \in \mathbf{triH}$,}
  \STATE{Append $\{``b_1\text{''}: \mathbf{lines}(b_1)\}$ to $tri\_b_1 \in \mathbf{triH}$.}
\ENDFOR
\RETURN{ $\mathbf{triN},\mathbf{triE}, \mathbf{triH}$}
\end{algorithmic}
\end{algorithm}

In each iteration of \cref{alg:triangular}, we perform a greedy search for permissible triangles connected to a base node $b_1$, i.e. all triangles for which the nodes fall below the specified voltage and degree thresholds.  We recursively collapse all such triangles into $b_1$ by removing the two associated nodes from $\mathbf{triN}$ and connecting all their lines to $b_1$, avoiding double and self lines. We perform this search until there are no permitted triangles which include $b_1$ and start the search again from a new base node. The base node from which we search for triangles is randomized upon each iteration. Thus, for each combination of voltage and degree threshold, we run an ensemble of experiments to find a distribution for our results. We plot the distribution of the degrees of nodes in $\mathbf{triN}$ over $10^3$ experiments in \cref{fig:tr_red_deg_dist}; for reference we include the degree distribution of nodes in $\mathbf{d2N}$. Note, while the triangular reduction produces nodes of degree at most 18, the reduction may \emph{lower the degree of any node} if it is connected to at least two nodes in a permissible reduction, as shown in \cref{fig:tr_deg_loss}. In \cref{fig:tr_red_deg_dist}, the newly apparent nodes of degree greater than 18 correspond to this phenomena, where various nodes of degree greater than 18 have been lost, and newly apparent nodes of degree above 18 are visible in the triangle reduced networks. 

\begin{figure}[t]\label{fig:tr_red_deg_dist}
\center
\includegraphics[width=.95\linewidth]{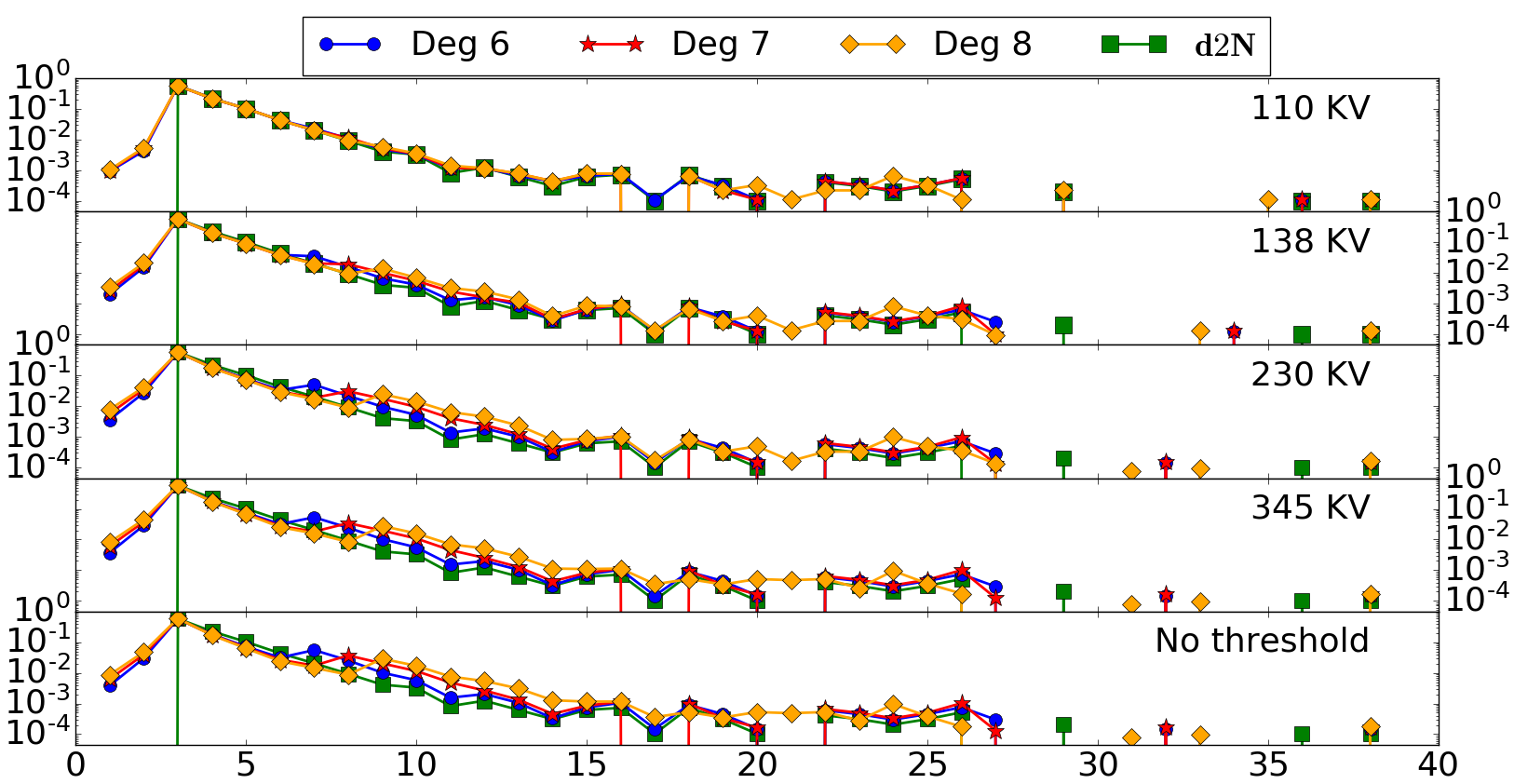}
\caption{Distribution of nodal degrees in $\mathbf{triN}$, for each parameter setting, versus the degree. \textbf{Top to bottom}: figures ascending in voltage threshold with degree thresholds plotted together.}
\end{figure}
\begin{figure}[ht]\label{fig:tr_red_deg_dist_wasserstein}
\center
\includegraphics[width=.95\linewidth]{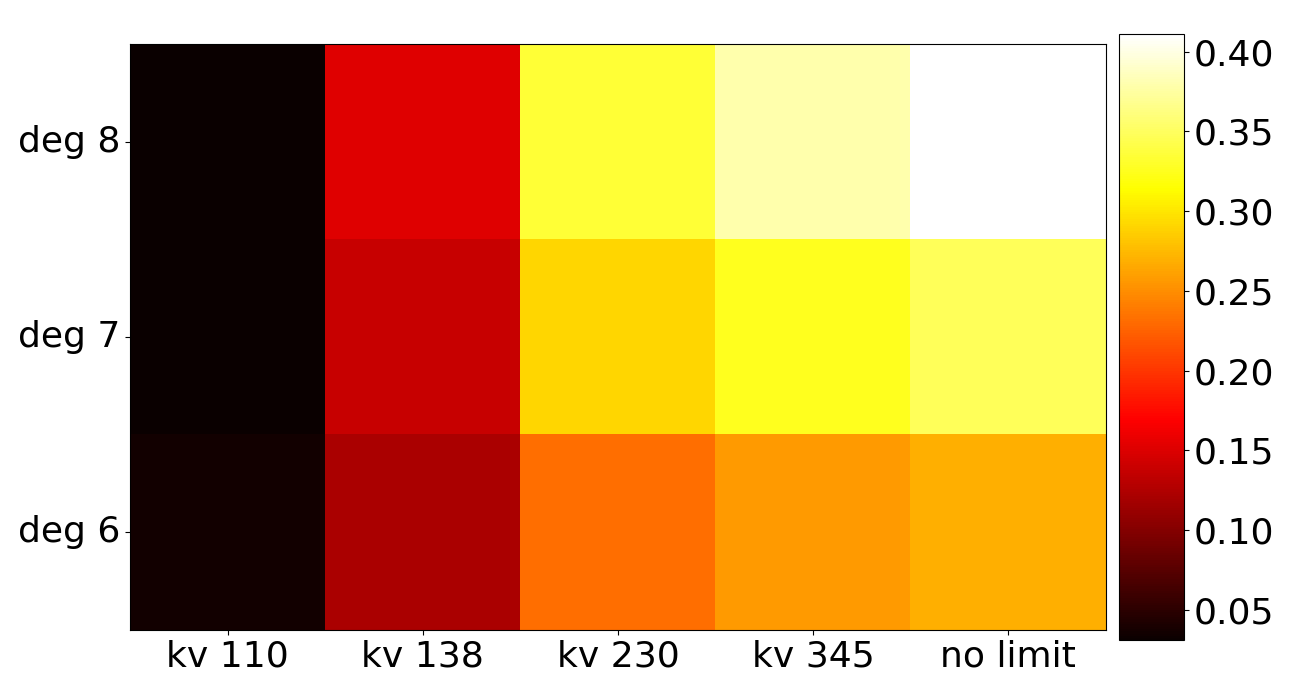}
\caption{The first Wasserstein distance, with L1 ground distance, between the triangle reduced degree distributions and the degree distribution for $\dtwo\bN$.}
\end{figure}

The smallest network produced by \cref{alg:triangular} has 5,560 nodes and 11,079 edges --- this is used as a reference for the possible limits of the triangular reduction, performed without a voltage threshold.  Even without the voltage threshold, the degree threshold in the reduction maintains the sparsity of the graph, which has a density of approximately $7.1690\times 10^{-4}$ in the smallest realization of the triangular reduction, pictured in \cref{fig:tr_red_deg_dist}.  In each of the degree threshold and voltage threshold settings, we additionally compute the first Wasserstein distance between the degree distribution for the greedy triangular reduction and the reference $\dtwo \bN$ degree distribution.  The distance between the triangular reduction distribution and the degree two reduction is shown in \cref{fig:tr_red_deg_dist_wasserstein}.   We note that, although the degree distribution for $\dtwo \bN$ differs significantly from the original network, the bias introduced is justified by the physical interpretation of collapsing generalized trees and edges.  

We wish, thus, to compare the degree distributions of the triangular reductions with that of $\dtwo \bN$ to determine to what extent the triangular reduction: (i) erroneously re-introduces degree one and degree nodes, and/or (ii) distorts the degree distribution of the meshy, densely connected sub-networks.  In particular, we see that there is little difference between the degree distribution of the triangular reduced network with a voltage threshold less than or equal to 138 KV, with each of the degree thresholds.  However, raising the voltage threshold to 230 KV and above, significant differences emerge in the degree distributions.  The upper threshold of a distance of $0.4$ can be equated with the work to move almost half the distribution one bin to the right.

\begin{figure}[h]\label{fig:tr_size_red_dist}
\center
\includegraphics[width=.95\linewidth]{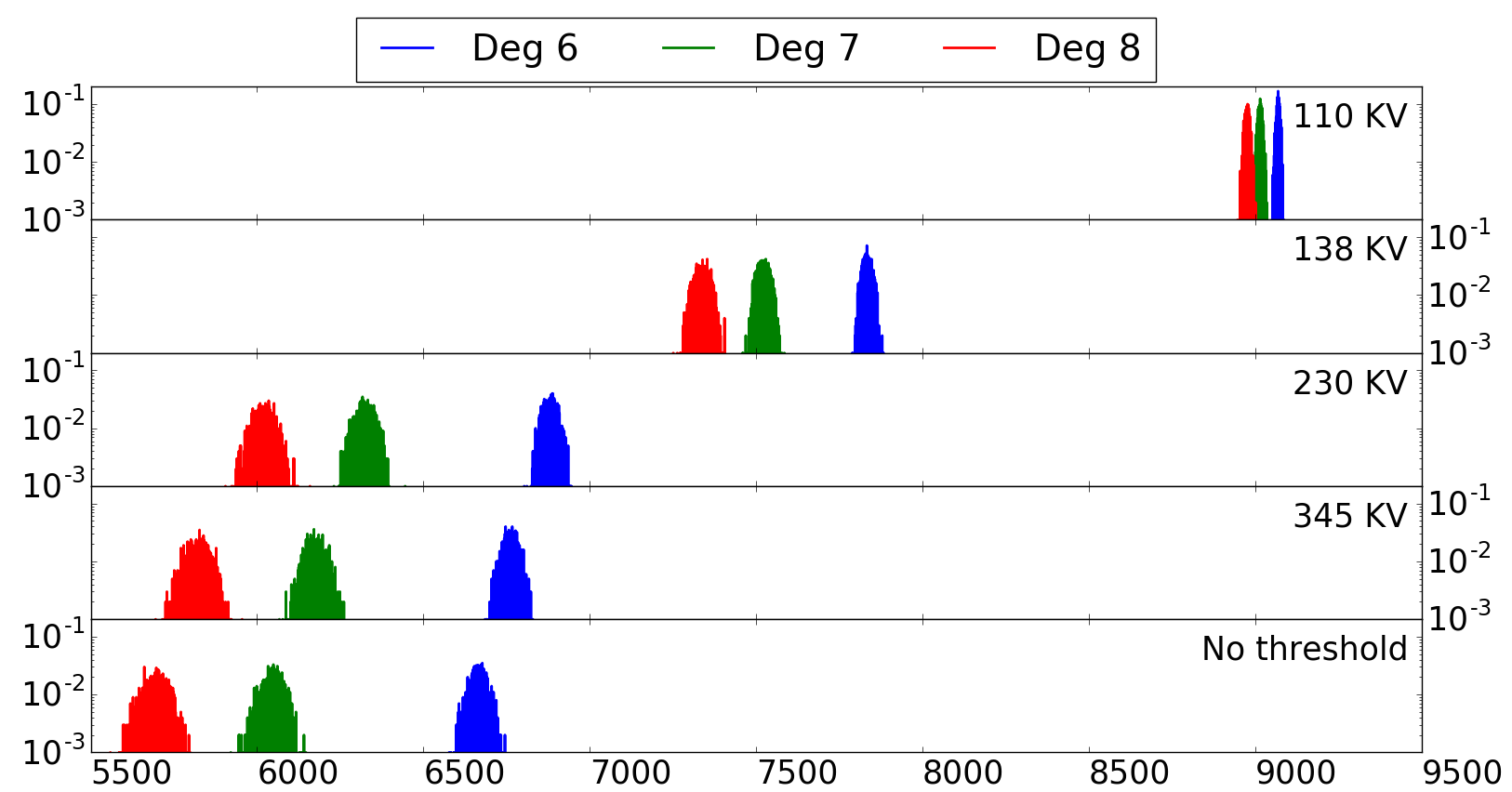}
\caption{Distribution of size of $\mathbf{triN}$ for each threshold setting.}
\end{figure}

We likewise see this behavior when we plot the distribution of \emph{size} of the reduced network $\mathbf{triN}$ in \cref{fig:tr_size_red_dist} with respect to the various threshold settings over 1000 initializations.  Sensitivity in the size of the reduced network to the voltage threshold emerges as we pass both from 110 KV to 138 KV, and from 138 KV to 230 KV thresholds respectively.  The distributions of the network size are all close and strongly peaked for the 110 KV threshold, indicating that few nodes in the distribution sub-network remain un-clustered after the degree one and degree two steps.  However, the dramatic reductions to network size passing to the 138 KV threshold indicates that the nodes of the distribution sub-network, and the substations connecting these to the transmission network (including super nodes which combine the two), possesses a loopy configuration that can be clustered by the triangular reduction for a significant gain.  This is also dynamically meaningful, as by construction, the reduction is only aggregating distribution sub-networks and transmission sub-stations with a similar dynamical interpretation to our earlier algorithms.  

Passing to the 230 KV threshold and above, there is once again a large reduction in the network size, where the loopy structure below the high voltage transmission network can be reduced significantly. The distributions of network size for voltage thresholds above 230 KV are more closely aligned, and are instead distinguished along their degree thresholds.  These additional large reductions, however, come at the cost of mixing distribution and transmission subnetworks, and distorting the degree distributions as indicated in \cref{fig:tr_red_deg_dist_wasserstein}.  Given the significant reductions produced under the voltage threshold of 138 KV, along with small distortions of the associated degree distributions, we conclude that the greedy triangular reduction can produce additional, physically meaningful clustering of the triangular meshes, formed at the distribution/ sub-station level of the network.  The smallest network produced under the triangular reduction, with $\mathbf{vThr}=138$ kv and $\mathbf{dThr}=8$ has 7,252 nodes, and 14,152 edges.  The mean nodal degree is approximately 3.90, with a standard deviation of 2.04, and a maximal nodal degree of 38.  The graph density is approximately $5.3826\times 10^{-4}$.

Although the topological procedure is intuitively clear, we have yet to discuss how to compute the admittances, currents, and power-flow for the triangle reduced network.  In the following, we will: (i) define the associated power-flow reductions and (ii) prove that the procedure in \cref{alg:triangular}, similar to the Kron reduction, produces a power-flow equivalent network which preserves graph paths.  This will lead us to our final analytical results, showing the ultimate consistency of our algorithms with the net power-flow, and the dynamics of the transient stability regime.

\begin{definition}\label{def:linearaggregation}
Let the loopy Laplacian $\bQ$ define an arbitrary, connected network of $n$ nodes satisfying \cref{hyp:inductive} and \cref{hyp:lossless}.  Moreover, let $\be_j \in\mathbb{R}^{n-3}$ be the $j$-th standard basis vector with all entries equal to zero except for a one in the $j$-th position.  Without loss of generality, let the nodes $\{b_{n-2}, b_{n-1}, b_n\}$ form a triangle.  Let $\alpha = \{1,\cdots,n-3\}$, we define the \textbf{linear aggregation of a triangle} formed by $\{b_{n-2}, b_{n-1}, b_n\}$ as follows:
\begin{enumerate}[(i)]
  \item the triangle reduced loopy Laplacian $\bQ^{\rm tri}$ is given by
  \begin{align}\label{eq:trinodaladmittance}
  \bQ^{\rm tri} &\triangleq 
  \begin{pmatrix}
  \bQ_{[\alpha,\alpha]} & \sum_{i=n-2}^n \sum_{j=1}^{n-3} \be_{j}Q_{i,j} \\
   \sum_{i=n-2}^n \sum_{j=1}^{n-3} \be_{j}^\T Q_{i,j} &\sum_{i,j=n-2}^n Q_{i,j}
  \end{pmatrix};
  \end{align}
  \item the reduced current vector is given as $\bC^{\rm tri} \triangleq 
  \begin{pmatrix}
  \bC_{[\alpha]}^\T & \sum_{j=n-2}^n C_j
  \end{pmatrix}^\T$;
  \item the reduced current balance equations and power-flow equations are defined as
  \begin{align}
  \bV^{\rm tri} &\triangleq  \(\bQ^{\rm tri}\)^{-1} \bC^{\rm tri} \label{eq:trredcurrent}\\
  \bS^{\rm tri} &\triangleq \bV^{\rm tri} \circ \thickbar{\bC}^{\rm tri} \label{eq:trredpower}
  \end{align}
  
\end{enumerate}  
\end{definition}
\begin{lemma}\label{lemma:welldefinedtriangle}
The linear aggregation of the nodes $\{b_{n-2}, b_{n-1}, b_n\}$ into a single node, as described in \cref{def:linearaggregation} satisfies the following:
\begin{enumerate}[(i)]
\item $\bQ^{\rm tri} \in \mathbb{C}^{(n-2) \times (n-2)}$ is an invertible, loopy Laplacian, such that equations \cref{eq:trredcurrent} and \cref{eq:trredpower} are well defined;
\item the network defined by $\bQ^{\rm tri}$ satisfies \cref{hyp:inductive} and \cref{hyp:lossless}; 
\item the triangle reduced adjacency matrix $\bA^{\rm tri}$ has $A^{\rm tri}_{i,n-2} \neq 0 $ if and only if $A_{i,k} \neq 0$ for some $k \in \{n-2,n-1,n\}$.
\end{enumerate}
\end{lemma}
\begin{proof}
In a linear circuit, we can equivalently combine parallel edges into a single edge simply by taking the sum of the parallel line admittances to be the line admittance for the single, reduced line. In particular, all lines that connect the reduced triangle to the external network are defined this way via equation \cref{eq:trinodaladmittance} --- the off diagonal elements of row $n-2$ of $\bQ^{\rm tri}$ are equal to the negative of the sum of all line admittances exterior to the triangle.  Notice that the term $\sum_{i,j=n-2}^n Q_{i,j}$ equals the sum of all elements in $\bQ_{(\alpha,\alpha)}$.  In particular, this sum cancels out all copies of the line admittances internal to the triangle, i.e., $A_{n-2,n-1}, A_{n-2,n}, A_{n-1,n}$, while leaving all other summands unaffected.  Thus, by \cref{eq:loopyelement}, the sum $\sum_{i,j=n-2}^n Q_{i,j}$ combines all of the line admittances exterior to the triangle $\{b_{n-2},b_{n-1},b_n\}$ and the sum of the shunt admittances for each node.  

The above shows that $\bQ^{\rm tri}$ is a well defined loopy Laplacian, and that the circuit between the reduced triangle and all other nodes is electrically equivalent to the unreduced network.  Similarly, modeling the loads within the triangle formed by nodes $\{b_{n-2},b_{n-1},b_n\}$ as shunt admittances, or self loops, we can equivalently define the new load in the aggregated triangle as follows.  Each of the loads within the three nodes $\{b_{n-2}, b_{n-1}, b_n\}$ become parallel self loops in the reduced circuit, as do the internal edges to the triangle $\{b_{n-2},b_{n-1}\},$ $\{b_{n-2},b_n\},$ and $\{b_{n-1},b_n\}$.  However, unlike the loads, the self loops corresponding to the internal edges $\{b_{n-2},b_{n-1}\},$ $\{b_{n-2},b_n\},$ and $\{b_{n-1},b_n\}$ both draw and re-inject power.  Therefore, the electrically equivalent circuit of self loops of the reduced triangle must cancel the admittances of the edges $\{b_{n-2},b_{n-1}\},$ $\{b_{n-2},b_n\},$ and $\{b_{n-1},b_n\}$.   The electrically equivalent self loop for the super node representing the reduced triangle thus has line admittances equal to the sum of the shunt admittances for the three nodes $\{b_{n-2},b_{n-1},b_n\}$.  By equation \cref{eq:adjacencyelement}, we see that the shunt admittance for the reduced triangle super node is equal to
\begin{align}\begin{split}
A^{\rm tri}_{n-2, n-2} &= -\left(\sum_{j=1}^n \sum_{i=n-2,\hspace{1mm} i\neq j }^n A_{i,j}\right) + \sum_{j=1}^n \sum_{i=n-2}^n A_{i,j} \\
& = \sum_{i=n-2}^n A_{ii}
\end{split}\end{align}
Thus the circuit within and without the reduced triangle are electrically equivalent to the unreduced network.

We take the lines in the triangle reduced network to be lossless as they are simply the combination of lines in the unreduced network.  Note that the currents for all nodes in the reduced network outside of the collapsed triangle, and the total current within the super node representing the reduced triangle, are preserved by construction, equaling those of the unreduced network.    Therefore, net power is preserved via equations \cref{eq:trredcurrent} and \cref{eq:trredpower} by construction, provided they are well defined.  Particularly, the inverse $\bQ^{\rm tri}$ adjusts the nodal voltages consistently in the reduced network according to Ohm's law, with the constraints of the preserved currents and the electrically equivalent circuit. 

Consider thus, the diagonal of $\bQ_{[\alpha,\alpha]}$ is dominant in $\bQ^{\rm tri}$, as for each row $i< n-2$, $Q^{\rm tri}_{i,n-2}$ is simply the sum of the column elements in positions $n-2, n-1$ and $n$ in $\bQ$.  Likewise, the diagonal element $Q^{\rm red}_{n-2,n-2}$ is dominant by construction, and the diagonal dominance of $\bQ$ implies the diagonal dominance of $\bQ^{\rm tri}$.  Let $\bA^{\rm tri}$ be defined by equation \cref{eq:adjacencyelement}.  We assume at least one element $A_{ii}\neq 0$, and this must also hold for $\bA^{\rm tri}$ by construction, verified by the relationship between the elements in equations \cref{eq:loopyelement} and \cref{eq:adjacencyelement}.  By construction, the associated graph is connected, and the elements of $\bA^{\rm tri}$ must be negative imaginary by the definition of equation \cref{eq:adjacencyelement}.  Therefore $\bA^{\rm tri}$ satisfies \cref{hyp:inductive} such that, by \cref{lemma:invertibleadjacency}, $\bQ^{\rm tri}$ is invertible.  This proves statement $(i)$ above, and thus statement $(ii)$.  Statement $(iii)$ above is trivial by \cref{def:linearaggregation}. 
\end{proof}

\begin{cor}\label{cor:recursive_triangle}
Let $b_1, b_2\in \mathbf{triN}$, the network reduced via \cref{alg:degree1}, \cref{alg:nodetoedge}, \cref{alg:collapselasso} and \cref{alg:triangular}.  There exists an edge between $b_1$ and $b_2$ if and only if there exists a path from $b_1$ to $b_2$ in $\{\bN, \bE\}$ such that all interior nodes in the path belong to $\bN \setminus \mathbf{triN}$. That is, \cref{alg:degree1}, \cref{alg:nodetoedge}, \cref{alg:collapselasso} and \cref{alg:triangular} preserve graph paths. 
\end{cor}
\begin{proof}
By \cref{cor:deg2graphpath}, we know that the statement is true for the degree two reduced network.  By the construction in \cref{def:linearaggregation}, the statement is trivial under a single iteration of the triangular reduction, and thus holds for the network given by $\mathbf{triN}$.
\end{proof}

\begin{theorem}
The sequential network reductions produced via \cref{alg:degree1}, \cref{alg:nodetoedge}, \cref{alg:collapselasso} and \cref{alg:triangular} produce a reduced order model which is power-flow equivalent, for a lossless, inductive, steady state network, and preserves graph paths.
\end{theorem}
\begin{proof}
\cref{cor:deg2kron} demonstrates that this holds for the degree two reduced network, and \cref{lemma:welldefinedtriangle} and \cref{cor:recursive_triangle} show that the triangular reduction can be iterated upon this network, preserving graph paths and maintaining the net power-flow equivalence.
\end{proof}

\section{Visualization of the reduced network}
\label{sec:visualization}

In this section, we discuss methods of graph visualization for the models produced by the degree two reductions and the triangular reductions.  Following Wong et. al. \cite{wong2009novel}, we choose to visualize the network based on its graph characteristics as in the GreenGrid visualization package. Rather than visualizing our network by the geographic information, a graph theoretic layout can better represent dynamical coupling and grid vulnerabilities. The classical force directed layout technique uses a spring and repulsion model where each node is a repelling body and the edges are represented by springs \cite{ulrik2001, kobourov2013}. Initial node positions are chosen randomly and the n-body problem is solved until the positions of nodes stabilize. This pseudo physical model can be utilized to represent electric grid physics by parameterizing spring lengths via the line admittances and node repulsion with the nominal voltages of nodes \cite{wong2009novel}.

We utilize the JavaScipt library vis.js \cite{visjs} to perform interactive visualizations. The default graph layout uses uniform spring lengths and repulsion parameters, and implementing a parametrization scheme that reflects the nodal voltage and line admittances in the \emph{reduced network models} is the subject of future work. The underlying ForceAtlas2 model \cite{jacomy2014} in vis.js is used to resolve the spring repulsion evolution. We produce a conceptual visualization of the clustering performed via \cref{alg:triangular} as follows: (i) first generate node positions for $\mathbf{triN}$, using the 138 degree voltage threshold and setting $\mathbf{dThr}=8$, and resolve the ForceAtlas2 model until node positions stabilize; (ii) fix these node locations and assign the initial position for every node in $\mathbf{d2N}$ as its clustered position in $\mathbf{triN}$; resolve the ForceAtlas2 model until node positions of $\mathbf{d2N}$ stabilize. \cref{fig:net_vis} demonstrates a realization of this de-clustering: the left hand plot shows the initial positions for the nodes in $\mathbf{d2N}$; the middle plot describes an intermediate point in their evolution as node positions are released; the right hand plot visualizes the stabilized $\mathbf{d2N}$ positions. This de-clustering visualization demonstrates how the degree threshold maintains qualitative graph features during the reduction.  Specifically, in the visualizations we see the emergence of densely connected coherent areas of nodes, sparsely connected by inter-area lines.

\begin{figure}[h]\label{fig:net_vis}
\center
\includegraphics[width=.8\linewidth]{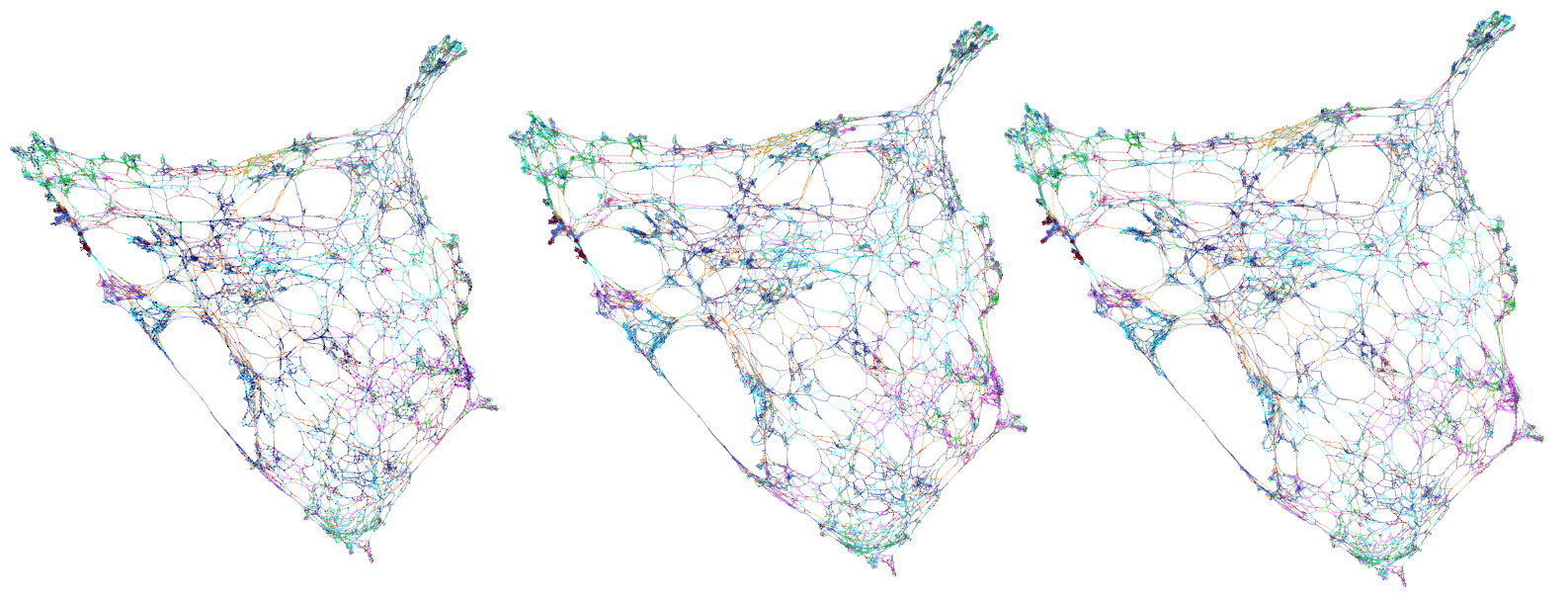}
\caption{\textbf{Left}: ForceAtlas2 initialized for $\mathbf{d2N}$, $\mathbf{d2E}$ with the clustered positions in $\mathbf{triN}$, $\mathbf{triE}$, using $\mathbf{vThr=8},\mathbf{dThr=8}$. \textbf{Middle}: the node positions are propagated by the ForceAtlas2 model. \textbf{Right}: positions stabilize.}
\end{figure}

\section{Conclusions}
\label{sec:conclusions}
Analysis of our test network demonstrates that our reductions meet the goals stated in \cref{sec:alg}. Firstly, our graph based approach to network reduction preserves network topological features such as graph paths and sparsity.  Although the degree one and degree two reductions fundamentally change the degree distribution of the original network, this bias is justified by the physical coherence of the nodes collapsed, and their associated design features.  Moving to the triangular reductions, we find parameter regimes with the 138 KV voltage threshold which make physically meaningful nodal aggregations, while preserving the degree distribution of the degree two reduced network, and the overall sparsity.  

Each of our algorithms are proven to produce a power-flow equivalent network, under the hypothesis of a lossless, inductive, steady state network, allowing for a physically meaningful qualitative analysis of synchronization, optimal power-flow and control in the transient stability regime.  Moreover, by a sequential, recursive design, our procedure allows a partial reconstruction of the full network from a sequence of intermediate reduced models with varying levels of resolution: efficient use of data structures allows the user to reconstruct sequential reductions and reintroduce complex network features. Finally, we demonstrate the potential for interactive visualization of the reduced model for qualitative study of network sensitivities. As an additional step, one may use the graph based visualization to represent the dynamical coupling in the reduced network, using the (clustered) nodal voltage to represent repulsion and (meta-)edge admittances to represent spring parameters \cite{wong2009novel}. Visualizing the reduced network this way preserves and even distinguishes major qualitative features of the original model, using the visually comprehensible reduced network.

\begin{figure}
\small
\center
\begin{tabular}[h]{|| l | l | l | l || }
\hline
$\rvert \mathbf{N} \rvert = 53,155$ & $\rvert \mathbf{d1N} \rvert = 32,891$  & $\rvert \mathbf{d2N}\rvert = 9,716$ & $\rvert \mathbf{triN}\rvert = 7,252$ \\
\hline
$\rvert \mathbf{E} \rvert = 63,832$ & $\rvert \mathbf{d1E} \rvert = 43,568$  & $\rvert \mathbf{d2E}\rvert = 18,700$ & $\rvert \mathbf{triE}\rvert = 14,152$ \\
\hline
\end{tabular}
\caption{Test network reduction summary, $\mathbf{triN}, \mathbf{triE}$ correspond to the smallest realization under the 138 kv voltage threshold and degree 8 threshold.}
\end{figure}

Although we have shown analytically that, under \emph{ideal conditions}, the power-flow in the reduced network will be equivalent to the full network, we have not yet performed dynamical simulations to test the limits of this equivalence.  Specifically, we have not treated the realistic scenarios of non-static voltages, stochasticity in the generation and loads, mid-term stability regimes or the learning problem for the dynamic swing equations \cite{lokhov2018online}, where we must estimate the damping and inertial parameters for each of the aggregated nodes in the reduced network.  Additionally, while our topological reductions intuitively appear to be consistent with other reduction methodologies such as slow coherency \cite{chow1985time,chow2013power}, we have yet to make a quantitative comparison of the methods, to determine in what ways these are complementary.  Each one of the above questions is relevant for our ultimate goal of designing computationally efficient reduced order models for online state and parameter estimation, and will be the subject of future work.

\appendix

\section{Data structures and inverting reductions}
\label{sec:data}
Allowing users to refine the reduced network structure is basic to our algorithm design. We expand in detail the data storage of generalized trees, edges and collapsed triangles. The recursion in \cref{alg:nodetoedge} and \cref{alg:collapselasso} implies that edge and generalized tree data structures can be multilayered, containing multiple levels of sub-edges or sub-trees. Proceeding from the bottom layer to the top, and from right to left within lists, one can recover the reverse sequence of mappings to reconstruct a node. An example interactive visualization is available in web browsers \cite{grudzienpersonalpage}, demonstrating the de-clustering performed in \cref{fig:net_vis}. We likewise release our reduction scripts and toy data describing the full and reduced network node and edge sets, with voltage information in an arbitrary, per unit representation \cite{grudziengithub2017}.

\subsection{Tree data}\label{sec:treedata}

Tree reductions are called by a field $t\_b_1$ where $b_1$ is the terminal node of the collapse in \cref{alg:degree1}.
Each field returns a list of arrays, each array corresponding to a branch collapsed to the root node $b_1$.
The first position of each array describes the end leaf of the branch and each subsequent position describes the shortest path in the network to the terminal node. \cref{fig:treecollapse} corresponds to the list
\begin{align}\label{eq:tree}
\mathbf{d1H}(t\_b_1) =& \left\{[b_2, b_1], [b_3, b_1],[b_6, b_5, b_4, b_1], [b_7, b_5, b_4, b_1]\right\},
\end{align}
where leaves are reintroduced by following the path described in the array.  Each leaf in the tree $t\_b_1$ can be re-introduced by its shortest path to $b_1$, described in equation \cref{eq:tree}. The terminal node $b_1$ has at least two lines connecting it to the remaining network.

\subsection{Edge data}

Let $b_2$ be a node of degree two, and suppose it is connected to $b_1$ and $b_3$. The basic mapping produced by \cref{alg:nodetoedge} takes $b_2$ to the edge $\{b_1,b_3\}$. We represent this map by the array $[b_1, b_2, b_3]$ where, without loss of generality, we assume that $b_1 < b_3$. Given such a sequence of mappings
\begin{align}\label{eq:edge}
\mathbf{d2H}(e\_b_1\_b_5) = \left\{[b_1,b_2,b_3], [b_1,b_3,b_4], [b_1,b_4,b_5]\right\}
\end{align}
we may reconstruct the original string by following the mappings from right to left in the list. Equation \cref{eq:edge} describes the line of nodes in \cref{fig:edgemap}. If $b_4$ is the terminal node for a generalized tree, \cref{alg:nodetoedge} stores the associated generalized tree data in the list at the point of the reduction. The list is thus given
\begin{align}\label{eq:edgetree}
 \mathbf{d2H}(e\_b_1\_b_5) = \left\{[b_1,b_1,b_3], [b_1,b_3,b_4], [b_1,b_4,b_5], \{``t\_b_4\text{''}: \text{``tree data ''}\} \right\}
\end{align}
Generalized trees embedded in an edge are reconstructed by reintroducing the terminal node from the edge data and reconstructing the generalized tree as described in \cref{sec:gentree}. An edge in $\mathbf{d2H}$ may contain an arbitrary length sequence of edge and tree reductions, possibly multilayered. Each meta-edge may therefore be represented in multiple ways by different orders of mappings, but each map can be inverted sequentially to reconstruct the original network, regardless of the order.

\subsubsection{Generalized tree maps}\label{sec:gentree}
\cref{lemma:gentree} demonstrates that a sparsely connected triangle is collapsed if and only if \cref{alg:nodetoedge} produces a degree one node. Generalized tree data, therefore, includes the sequence of nodes mapped to edges which precipitate collapse of the triangle. The field $t\_b_n$ corresponds to the node, $b_n$, that the generalized tree has been collapsed to. Suppose as in \cref{fig:gentree}, mapping $b_1$ to the edge $\{b_2,b_3\}$ produces a degree one node in $\mathbf{d2N}$. \cref{alg:collapselasso} collapses degree one nodes recursively until every node is again at least degree two. Let $b_k$ be the terminal node of this collapse, then \cref{alg:collapselasso} stores a hashable map as the first entry of $t\_b_k \in \mathbf{d2H}$, followed by the array with the path from $b_2$ to the terminal node
\begin{align}
\mathbf{d2H}\left(t\_{b_k}\right) = \{\{ ``e\_b_2\_b_3\text{''} : [b_2, b_1, b_3]\}, [b_2, b_3, \cdots, b_k]\}
\end{align}
In general, the edge data precipitating the collapse of the sparsely connected triangle can be of arbitrary length and contain multiple layers.

\subsection{Triangular reductions}

Due to the more arbitrary nature of the mapping  \cref{alg:triangular}, we take a simple approach to track the reductions. The field $tri\_b_1 \in \mathbf{triH}$ corresponds to a list where each entry is a hashable map of the form $\{ ``b_j\text{''} : \mathbf{lines}(b_j)\}$. The value $\mathbf{lines}(b_j)$ is the list of lines associated to $b_j$ in $\mathbf{d2E}$. In this way, one can reintroduce a node from a collapsed triangular configuration by writing the node $b_j$ into $\mathbf{triN}$ and reconnecting this node with the appropriate edges from $\mathbf{d2E}$, while removing these edges from $b_1$ if the lines were formed uniquely by joining $b_j$ to the cluster.

\bibliographystyle{siamplain}
\bibliography{references}
\end{document}